\documentclass{article}
\usepackage[english]{babel}
\usepackage[dvipsnames]{xcolor}
\usepackage[letterpaper,top=2cm,bottom=2cm,left=3cm,right=3cm,marginparwidth=1.75cm]{geometry}
\usepackage{amsmath}
\usepackage{graphicx}
\usepackage{amssymb}
\usepackage{amsthm}
\usepackage{proof}

\usepackage{tikz}
\usetikzlibrary{decorations.markings}

\usepackage[colorlinks=true, allcolors=blue]{hyperref}

\newtheorem{prop}{Proposition}
\newtheorem{thm}{Theorem}

\RequirePackage{verbatim}

\makeatletter
\newwrite\bibinl@out
\newenvironment{bibtex}[1][\jobname]{%
  \immediate\openout\bibinl@out #1.bib
  \immediate\write\bibinl@out{\@percentchar generated from `\jobname' starting line \the\inputlineno^^J}%
  \def\verbatim@processline{\immediate\write\bibinl@out{\the\verbatim@line}}%
  \@bsphack\let\do\@makeother\dospecials\catcode`\^^M\active\verbatim@start
}%
{\immediate\closeout\bibinl@out\@esphack}
\makeatother

\providecommand{\href}[2]{#2}

\makeatletter
\def\mr@ignsp#1 {\ifx\:#1\@empty\else #1\expandafter\mr@ignsp\fi}%
\newcommand{\multiref}[1]{\begingroup
\xdef\mr@no@sparg{\expandafter\mr@ignsp#1 \: }%
\def\mr@comma{}%
\@for\mr@refs:=\mr@no@sparg\do{\mr@comma\def\mr@comma{,}\ref{\mr@refs}}%
\endgroup}
\renewcommand{\eqref}[1]{(\multiref{#1})}
\makeatother

\makeatletter
\newcommand{\namedref}[2]{\hyperref[#2]{#1~\ref*{#2}}}
\newcommand{\secref}{\@ifstar{\namedref{Section}}{\namedref{Sec.}}}
\newcommand{\appref}{\@ifstar{\namedref{Appendix}}{\namedref{App.}}}
\newcommand{\tabref}{\@ifstar{\namedref{Table}}{\namedref{Tab.}}}
\newcommand{\figref}{\@ifstar{\namedref{Figure}}{\namedref{Fig.}}}
\makeatother

\title{All $4\times4$ solutions of the quantum Yang--Baxter equation}
\author{Marius de Leeuw$^1$ and Vera Posch$^1$}

\begin{document}

\begin{titlepage}

\centerline{\large \bf Solving the $4\times 4$ quantum Yang--Baxter equation}
\vskip 2 cm

\centerline{{\bf Marius de Leeuw and Vera Posch}  }

\vskip 1.1cm

\begin{center}
\sl Hamilton Mathematics Institute\\
~School of Mathematics\\
~Trinity College Dublin, 
~Dublin, Ireland\\ \vspace{0.2cm}
\end{center}
\vskip 0.7cm

\centerline{\small\tt \{deleeuwm,poschv\}@tcd.ie}

\vskip 1.7cm \centerline{\bf Abstract} \vskip 0.2cm \noindent
In this paper, we complete the classification of $4\times4$ analytic solutions of the Yang--Baxter equation. Regular solutions were recently classified and in this paper we find the remaining non-regular solutions. We present several new solutions, then consider regular and non-regular Lax operators and study their relation to the quantum Yang--Baxter equation.  We show that for regular solutions there is a correspondence, which is lost in the non-regular case. In particular, we find non-regular Lax operators whose $R$-matrix from the fundamental commutation relations is regular but does not satisfy the Yang-Baxter equation.  These $R$-matrices satisfy a modified Yang--Baxter equation instead.

\end{titlepage}


\section{Introduction}

In this paper, we solve the Yang Baxter Equation (YBE). This equation was named by Faddeev and Takhtajan after Yang and Baxter \cite{TakFaddeev}. Yang and Baxter encountered this equation in different contexts. Yang studied a one-dimensional many-body quantum mechanical scattering \cite{Phys} and Baxter consider the eight vertex model in statistical physics \cite{BAXTER}. In fact, this equation arises in many different settings in various fields in physics and mathematics. For a brief (historical) overview into the different fields in which Yang-Baxter equation arises, see for instance \cite{perk2006yangbaxterequations} or the book by Jimbo \cite{jimbo1990yang}. 

Today, the YBE is central to quantum integrability and more interesting and relevant to mathematical physics and applications than ever.
This paper concerns itself with solving the YBE and classifying its solutions. In general, this is a very daunting challenge. The Yang--Baxter equation is a matrix equation on an object called the $R$-matrix and takes the form
\begin{align}
\label{YBE}
R_{12}(u,v)R_{13}(u,w)R_{23}(v,w)=R_{23}(v,w)R_{13}(u,w)R_{12}(u,v)
\end{align}
with $R:V\otimes V \rightarrow V\otimes V$,  $R_{12}=R\otimes \mathbf{1}$ and $u,v\in \mathbf{C}$. The entries of the $R$-matrix are functions depending on spectral parameters. Hence, the Yang-Baxter equation describes a coupled set of cubic functional equations. Solving this is clearly not a simple task. The Hilbert space $V$ is usually chosen to be $\mathbf{C}^n$ and in this paper we will restrict to $\mathbf{C}^2$. 

When the local Hilbert space is two-dimensional, classification of solutions to Baxter's equations was initiated in \cite{Krichever1981BaxtersEA} using algebraic geometry.  It shows that the most general spectral-parameter dependent solution of the Yang-Baxter equation is given by the eight-vertex R-matrix under the assumption that the curve on which the spectral parameters live is elliptic. More degenerate cases have been studied in subsequent work, such as \cite{Dragovich}.

A special case of the Yang--Baxter equation \eqref{YBE} is where all spectral parameters are taken to be equal $u=v=w$. These correspond to so-called constant solutions. In the case of $V=\mathbf{C}^2$ all constant solutions were classified and listed in \cite{HIETARINTA}. The best known solution is the permutation matrix $P$, which acts as a permutation of vector spaces in a tensor product. The other constant solutions have only recently been picked up again by \cite{Maity:2024vtr}, where they discuss interesting algebraic properties described by these matrices.

In this paper we are interested in non-constant solutions and differentiable solutions. This is mainly motivated by the fact that in the standard construction, the conserved charges are defined using derivatives of the $R$-matrix\footnote{One could also define a family of conserved charges by evaluating the transfer matrix at different values of the spectral parameter which circumvents the differentiability constraint, but we will not consider this in this paper.}. Hence, we consider $R$-matrices that are analytic in the spectral parameters $u$ and $v$
\begin{align}
\label{reg-exp}
    R(u,v) = R^{(0)}+u R^{(1)}+v R^{(2)}+ uv R^{(3)}+ u^2 R^{(4)}+\dots
\end{align}
Due to the fact that the Yang--Baxter equation is not sensitive to the normalisation of $R$, we can always do this as long as the entries are meromorphic functions. In principle, we can solve the Yang--Baxter equation order by order in the spectral parameters. At zeroth order, we find the constant Yang--Baxter equation. Therefore $R^{(0)}$ needs to be one of the constant solutions found in \cite{HIETARINTA}, which we list in Appendix \ref{constantsolutions} for completeness. 

The most studied set of solutions are the so-called regular solutions where $R^{(0)} = P$. These correspond to spin chains with nearest neighbour interactions and usually correspond to $R$-matrices that appear as scattering matrices in integrable quantum field theories. Traditionally these types of solutions have been studied using symmetry arguments \cite{Drinfeld:1985rx,Jimbo:1985zk} or more direct approaches \cite{Vieira_2018} . More recently an approach based on the boost operator was set-up \cite{de_Leeuw_2019,deLeeuw:2020ahe,de_Leeuw_2021}, which resulted in the full classification for $4\times4$ regular solutions of the Yang--Baxter equation \cite{Corcoran:2023zax}.

So, in order to complete the classification of $4\times4$ analytic solutions  we need to consider the remaining constant cases which give non-regular $R$-matrices. In this paper we discuss how these can be extended to non-constant solutions. A first attempt at this was put forward in  \cite{Garkun:2024jnp} using a new algorithmic approach to construct solutions . In this paper we are completing the classification $4\times4$ non-regular analytic solutions of the quantum Yang--Baxter equation. 

With classification we mean that we list all \textit{analytic} solutions to the quantum Yang-Baxter equation with complex spectral parameters up to the following identifications 
\begin{enumerate}
\item Reparameterisations and redefinitions of the (spectral) parameters and free functions
\item Normalisation of the $R$-matrix $R\rightarrow f(u,v) R$
\item Transposition $R\rightarrow R^T$ and parity $R_{ij} \rightarrow R_{ji}$
\item Local basis transformations $R_{ij} \rightarrow V_{i} V_{j} R_{ij} V_{i}^{-1} V_{j}^{-1}$
\end{enumerate}
All of these transformations trivially map between solutions of the Yang-Baxter equation and in our list we will choose one representative of the equivalence class generated by these transformations. In particular, some solutions will depend on free functions $f,g,\ldots$ and we take solutions to be in the same class for any choice of these functions. Our methods are explained and illustrated by an example in Section \ref{Method}. In section \ref{Non-constant Models} all the models have been listed under the above identifications. 

The existence of a quantum integrable model is normally understood as the existence of a tower of commuting conserved charges. These are generated by a Lax operator $L$ which solves the fundamental commutation relations \cite{Faddeev:1987ih,Faddeev:1996iy} 
\begin{align}
R_{12}(u,v)L_{13}(u)L_{23}(v)=L_{23}(v)L_{13}(u)R_{12}(u,v)
\end{align}
By identifying the Lax operator with the $R$-matrix via $L(u)\sim R(u,0)$ this can be seen as a special case of the YBE. Indeed we show that for regular solutions this identification is always possible, so that any integrable model described by a regular Lax operator corresponds to a solution of the Yang--Baxter equation.

However, as we discuss in Section \ref{R-matrices and Lax operators}, this breaks down the non-regular setting. For an integrable model described by a non-regular Lax operator, we find that the fundamental commutation relations can be solved by $R$-matrices that do \textit{not} satisfy the Yang-Baxter equation. Instead they satisfy a modified Yang Baxter equation
%
%
We discuss two explicit examples to illustrate this.

\section{Method}
\label{Method}

In this section, we will explain the framework that we use to lift a constant solution of the Yang--Baxter equation to a solution that depends on spectral parameters. We assume that the spectral parameters are complex numbers, and we will also assume that the entries of the $R$-matrix admits a Laurent series expansion in terms of the spectral parameters. Since we can identify the solution up to overall normalization this means that we can assume that $R$ is analytic and admits a Taylor expansion by multiplying the $R$-matrix by an overall factor that cancels all poles. This condition is also desirable from a physical point of view, where one would define the conserved charges of the corresponding integrable models as logarithmic derivatives of the transfer matrix.

\subsection{Framework}

The key to our derivation is to expand around the known constant solutions. Let us show that at the point of coinciding spectral parameters, the $R$-matrix reduces to one of the constant solutions. This fixes the form of the expansion.

\begin{prop}\label{prop1}
Let $R(u,v)$ be a analytic solution of the Yang-Baxter equation. Then, we can write
\begin{align}\label{eq:Rexp}
R(u,v) = R^{(0)} (u_+) + u_- R^{(1)} (u_+) + \frac{u_-^2}{2} R^{(2)} (u_+) +\ldots,
\end{align}
where $u_- = u-v$ and $u_+ = \frac{u+v}{2}$ and $R^{(0)} (u_+) $ is one of the $R$-matrices from \cite{HIETARINTA} where we replace all constants with analytic functions $a\mapsto a(u_+)$.
\end{prop}

\begin{proof}
Since $R$ is analytic in $u,v$, it will also be analytic in the variable $u_\pm$. In, particular it admits a power series of the form \eqref{eq:Rexp}. Let us now consider the Yang-Baxter equation and set all the spectral parameters equal. We find
\begin{align}
R_{12}(u,u)R_{13}(u,u)R_{23}(u,u)=R_{23}(u,u)R_{13}(u,u)R_{12}(u,u).
\end{align}
This is exactly the constant Yang-Baxter equation whose solutions where classified in \cite{HIETARINTA}. However, this Yang-Baxter equation is satisfied by $R^{(0)}(u_+)$ and hence all constants that appear in \cite{HIETARINTA} need to be promoted to functions. 
\end{proof}

By expanding the Yang-Baxter equation around the points where two of the spectral parameters coincide, we can then derive the following necessary conditions that $R$ needs to satisfy:
\begin{equation}
    \begin{aligned}
         R_{12}^{(0)}(u_+) R_{13}(u,v)R_{23}(u,v)=R_{23}(u,v)R_{13}(u,v)R_{12}^{(0)}(u_+),\\ 
         R_{12}(u,v)R_{13}^{(0)}(u_+)R_{23}(v,u)=R_{23}(v,u)R_{13}^{(0)}(u_+)R_{12}(u,v),\\ 
         R_{12}(u,v)R_{13}(u,v)R_{23}^{(0)}(u_+)=R_{23}^{(0)}(u_+) R_{13}(u,v)R_{12}(u,v).
    \end{aligned}
    \label{3equ}
\end{equation}
Plugging \eqref{eq:Rexp} into \eqref{3equ} and letting $u\rightarrow v$ we can recursively solve for every order of $u-v$ independently. 
In particular, at leading order we find the constant Yang--Baxter equation again. All other orders give the same linear expression for the highest order contribution plus linear or quadratic corrections to the lower order contributions.
For example, at the next order we find the following equations
\begin{align}
\label{R(1)}
R_{12}^{(0)}  R^{(0)}_{13}R^{(1)}_{23}+  R_{12}^{(0)}  R^{(1)}_{13}R^{(0)}_{23}&=R_{23}^{(0)}R_{13}^{(1)}R_{12}^{(0)}+R_{23}^{(1)}R_{13}^{(0)}R_{12}^{(0)},\\ 
R_{12}^{(1)}R_{13}^{(0)}R_{23}^{(0)}- R_{12}^{(0)}R_{13}^{(0)}R_{23}^{(1)}&= R_{23}^{(0)}R_{13}^{(0)}R_{12}^{(1)}- R_{23}^{(1)}R_{13}^{(0)}R_{12}^{(0)},\\ 
R_{12}^{(1)}R_{13}^{(0)}R_{23}^{(0)} + R_{12}^{(0)}R_{13}^{(1)}R_{23}^{(0)}&=R_{23}^{(0)}R_{13}^{(1)}R_{12}^{(0)} + R_{23}^{(0)}R_{13}^{(0)}R_{12}^{(1)}, \label{R(1)C}
\end{align}
where we have omitted the explicit spectral dependence since all the matrices $R^{(i)}_{ab}$ depend on the same spectral parameter. These equations are not all independent as the first two add up to give the third. Notice that for regular $R$-matrices these linear equations are automatically satisfied, but for non-regular $R$-matrices these give non-trivial constraints. We can also see that $R^{(1)} = R^{(0)}$ is automatically a solution. This degree of freedom corresponds to an overall rescaling of our $R$-matrix 
\begin{align}
&R \mapsto (1+a (u-v)) R  && \Rightarrow &(R^0,R^1)\rightarrow (R^0, R^1 +a R^0)
\end{align}
Hence we can set the coefficient proportional to $R^{(0)}$ to 0 by choosing an appropriate normalisation of our $R$-matrix.

As an example of what conditions \eqref{R(1)}-\eqref{R(1)C}  may impose, let us take $R^{(0)} = 1$. In this case these equations give that
\begin{align}
[R_{ij}^{(1)} , R_{kl}^{(1)}]=0,
\end{align}
for all $i,j, k,l$. This basically implies that $R^{(1)}$ needs to be diagonal. Expanding further we get a coupled set of equations that we can solve perturbatively for each choice of $R^{(0)}$. This gives us a set of necessary conditions on the coefficients from the expansion of the $R$-matrix. What then remains is substituting the solution to the above constraints into the Yang--Baxter equation and solve for the remaining degrees of freedom. 

\subsection{Examples}

In order to illustrate our method, let us work out two examples to show the different steps that are needed.  We will start from two constant solutions (which can be found in Appendix \ref{constantsolutions}) as our starting point, namely $R^H$ and $R^A$.

\paragraph{Example $R^H$}

We take the expansion equation \eqref{eq:Rexp} and then solve equations \eqref{R(1)}-\eqref{R(1)C} for the elements of the unknown matrix $R^{(1)}$. The components of $R^{(1)}$ are free functions $f^{(1)}_{i}(u+v)$ that depend on the sum of the spectral parameters. Equations \eqref{R(1)}-\eqref{R(1)C}   then give rise to a set of linear equations in the free functions $f^{(1)}_i$ which can be solved. In this case, we find a solution with a single degree of freedom, which we call $f^{(1)}$ and we get
\begin{align}
    R^{H,(1)}= 
    f^{(1)}
    \begin{pmatrix}
        \cdot & \cdot & \cdot & 1\\
        \cdot & \cdot &1 &  \cdot \\
        \cdot & 1& \cdot & \cdot  \\
       -1& \cdot & \cdot &  \cdot \\
    \end{pmatrix}.
\end{align}
Plugging this back in our expansion(\ref{eq:Rexp}) we can then expand \eqref{3equ} to second order which leads to a system of equations linear in the components of $R^{H,(2)}$ and quadratic in $R^{H,(1)}$. The elements of $R^{H,(2)}$ are expressed in terms of $f^{(1)}$. To be more explicit we find:
\begin{align}
    R^{H,(2)}=
   \frac{1}{2} (  f^{(1)})^2  \begin{pmatrix}
        \cdot & \cdot & \cdot & 1\\
        \cdot & \cdot &1  &  \cdot \\
        \cdot & 1& \cdot & \cdot  \\
       -1& \cdot & \cdot &  \cdot \\
    \end{pmatrix}.
\end{align}
Going to the next higher orders this is a pattern that repeats. $R^{H,(n)}$ will be of the same structure as the previous ones, and it will depend on the elements of the previous $R^{H,(i)}$ ($i<n$). 

At this point, we can recursively show that the corresponding solution to the Yang--Baxter equation should take the form
\begin{align}
\mathcal{R}^H (u,v)= R^H + (f(u,v)-1)  \begin{pmatrix}
        \cdot & \cdot & \cdot & 1\\
        \cdot & \cdot &1 &  \cdot \\
        \cdot & 1& \cdot & \cdot  \\
       -1& \cdot & \cdot &  \cdot \\
    \end{pmatrix}.
\end{align}
 Plugging it into the Yang--Baxter equation we find that it is only satisfied if the functions $f(u,v)$ satisfy the following relation:
 \begin{align}\label{eq:fRH}
 f(u,w)=f(u,v) f(v,w).
 \end{align}
 This can be solved by expanding \eqref{eq:fRH} around $v=u$ and yields that 
\begin{align}
f(u,v) \equiv e^{F(u)-F(v)},
\end{align}
for some holomorphic function $F$. We can then use reparameterization invariance to simply define our spectral parameters as $F(u)\rightarrow u$ and we find our final answer:
\begin{align}
    \mathcal{R}(u,v)= \begin{pmatrix}
        1 & \cdot & \cdot & e^{u-v}\\
        \cdot & 1 &e^{u-v} &  \cdot \\
        \cdot &e^{u-v}& 1 & \cdot  \\
       -e^{u-v}& \cdot & \cdot &  1 \\
    \end{pmatrix},
\end{align}
which coincides with $\mathcal{R}^G$ from our classification below. This $R$-matrix is of difference form and was listed in \cite{Garkun:2024jnp}.

\paragraph{Example $R^A$}
Let us now illustrate a case which is less straightforward by considering $R^A$. We start by considering \eqref{R(1)}, but since $R^A$ depends on three free functions $p,q,s$, we find that we need to distinguish several cases. For instance the space of solutions for $R^{(1)}$ will be different for some special values, such as \textit{e.g.} $p=0=q$. These different cases will lead to different models that have to be studied separately. 

In the remainder of the example, let us continue by selecting the case where $p=q=1$ and $s=-1$.  This leads again to a unique solution to \eqref{R(1)} given by
\begin{align}
\label{RA1}
    R^{A,(1)}(u_+)= \begin{pmatrix}
        \cdot & \cdot & \cdot & f^{(1)}_1(u_+)\\
        \cdot & f^{(1)}_2(u_+) &f^{(1)}_3(u_+) &  \cdot \\
        \cdot &f^{(1)}_4(u_+)& f^{(1)}_5(u_+) & \cdot  \\
        f^{(1)}_6(u_+)& \cdot & \cdot &  f^{(1)}_7(u_+) \\
    \end{pmatrix}.
\end{align}
Plugging this back into our expansion and solving for the second order $R^{(2)}$ we find that  
\begin{align}
\label{RA2}
    R^{A,(2)}(u_+)= \begin{pmatrix}
        \cdot & \cdot & \cdot & f^{(2)}_1(u_+)\\
        \cdot & f^{(2)}_2(u_+) &f^{(2)}_3(u_+) &  \cdot \\
        \cdot &f^{(2)}_4(u_+)& f^{(2)}_5(u_+) & \cdot  \\
        f^{(2)}_6(u_+)& \cdot & \cdot &  f^{(2)}_7(u_+) \\
    \end{pmatrix}.
\end{align}
However, we also find several quadratic equations for the first order free functions $f^{(1)}_i$. These can be solved and we find three different solutions
\begin{align}
    R^{A,(1,1)}(u_+)&= \begin{pmatrix}
        \cdot & \cdot & \cdot & f^{(1)}_1(u_+)\\
        \cdot &\cdot &f^{(1)}_2(u_+) &  \cdot \\
        \cdot &f^{(1)}_2(u_+)&\cdot & \cdot  \\
        f^{(1)}_4(u_+)& \cdot & \cdot &  \cdot \\
    \end{pmatrix},\\
      R^{A,(1,2)}(u_+)&= \begin{pmatrix}
        \cdot & \cdot & \cdot & f^{(1)}(u_+)\\
        \cdot &\cdot& \cdot &  \cdot \\
        \cdot &\cdot& \cdot & \cdot  \\
       \cdot& \cdot & \cdot &  \cdot\\
    \end{pmatrix},
\qquad
 R^{A,(1,3)}(u_+)=   \begin{pmatrix}
        \cdot & \cdot & \cdot & \cdot\\
        \cdot & \cdot &\cdot &  \cdot \\
        \cdot &\cdot& \cdot& \cdot  \\
        f^{(1)}(u_+)& \cdot & \cdot &  \cdot\\
    \end{pmatrix}.
\end{align}
We see that $R^{A,(1,3)}$ and $R^{A,(1,2)}$ are related by transposition, so they do not describe independent models. Hence, in this case there are two possible paths that describe independent solution.

Let us continue with $ R^{A(1,2)}$. Repeating the above procedure we find that the structure recursively repeats itself and leads to the $R$-matrix $ \mathcal{R}^E$ from our classification 
\begin{align}
      \mathcal{R}^E(u,v)= \begin{pmatrix}
       1& \cdot & \cdot & f(u,v)\\
        \cdot &1& \cdot &  \cdot \\
        \cdot &\cdot& 1& \cdot  \\
       \cdot& \cdot & \cdot &  -1\\
    \end{pmatrix}.
\end{align}
This is easily checked to be a solution of the Yang--Baxter equation. Note that in this example we do not find any conditions coming from the YBE on the free function $f(u,v)$. The class of solutions, summarized as model $ \mathcal{R}^E(u,v)$, therefore contains all matrices of the above form, with \textit{any} holomorphic function $f(u,v)$.

\section{Non-regular solutions of the Yang--Baxter equation}
\label{Non-constant Models}

We now present the remaining solutions to the Yang--Baxter equation and we discuss which constant solutions they follow from. \textit{Note that we have used invariance under similarity transformations to present the most compact form.} One can use reparameterizations, redefinitions of the free functions,  similarity transformations and discrete transformations such as transposition to generate more solutions, but those are dependent on the ones listed in this section.

For a complete overview of which constant solutions leads to which analytic solutions we refer to Appendix \ref{completeR}. There we list all constant solutions, including special cases and discuss how they extend to an analytic solution. In what follows we list generic functions as $f,g,h$ and include the explicit functional dependence. Constants are typically denoted by $a,b,c,\ldots$ and do not have any functional dependence.

\subsection{Rank 4}

We first  list the invertible solutions.

\paragraph{Diagonal solution}
There is the obvious diagonal solution to the Yang--Baxter equation
\begin{align}
    \mathcal{R}^{A}(u,v)=\begin{pmatrix}
        f_1(u,v) &\cdot&\cdot&\cdot\\
        \cdot&f_2(u,v)&\cdot&\cdot\\
        \cdot&\cdot&f_3(u,v)&\cdot\\
        \cdot&\cdot&\cdot&f_4(u,v)
    \end{pmatrix}.
\end{align}

This model can be found from any of the constant solutions that can be reduced to a diagonal one. This matrix is invertible if all the functions along the diagonal are non-zero. Of course, even in the case where some functions vanish, this still gives a solution of the Yang--Baxter equation, albeit of rank lower than 4. Notice that solutions in this class are described by four free holomorphic functions and the different choices of functions can lead to different physical properties, but they all belong to the same class of Yang-Baxter solutions.

\paragraph{Off-Diagonal solution}
There is an anti-diagonal solution to the Yang--Baxter equation
\begin{align}
    \mathcal{R}^{B}(u,v)=\begin{pmatrix}
       \cdot&\cdot&\cdot& b\; g(u,v) \\
        \cdot&\cdot&f(u,v)&\cdot\\
        \cdot&f(u,v)&\cdot&\cdot\\
       a\; g(u,v)&\cdot&\cdot&\cdot
    \end{pmatrix}.
\end{align}
With a similarity transformation you can set $a=b$.

\paragraph{XY type solution}
There is a solution which takes the form of a XY type Hamiltonian
\begin{align}
    \mathcal{R}^{C}(u,v)= f(u,v) \sigma^i \otimes \sigma^i + g(u,v) \sigma^j \otimes \sigma^j,
\end{align}
with $i\neq j$ ,where $i,j \in \{x,y,z\}$ and $\sigma^i$ are the corresponding Pauli matrices. Setting $i=j$ would be equivalent to the diagonal case. There is one further special case in this class
\begin{align}
    \mathcal{R}^{D}(u,v)= f(u,v) \sigma^z \otimes \sigma^z + g(u,v) \sigma^\pm \otimes \sigma^\pm.
\end{align}
The $\pm$ sign flips under transposition and hence this is only one independent solution. The case where we have $i,j = \pm$ also solves the Yang--Baxter equation but it is of lower rank and will be discussed below. These types of solutions have also been found in \cite{Padmanabhan:2024zma}, see \textit{e.g.}  eqn (4.1) and eqn (4.5).

\paragraph{Upper triangular}
We find several models where the $R$-matrix has an upper-triangular structure
\begin{align}
&   \mathcal{R}^E(u,v)=\begin{pmatrix}
        k^2&\cdot&\cdot&\cdot\\
        \cdot&k q(v)&e^{f(u)-f(v)}(k^2-q(v)p(u))&\cdot\\
        \cdot&\cdot&k p(u)&\cdot\\
        \cdot&\cdot&\cdot&-p(u)q(v)
    \end{pmatrix},
    \\
&  \mathcal{R}^{E'}(u,v)=\begin{pmatrix}
        k^2&\cdot&\cdot&\cdot\\
        \cdot&k q&e^{f(u)-f(v)}(k^2-qp)&\cdot\\
        \cdot&\cdot&k p&\cdot\\
        \cdot&\cdot&\cdot&k^2
    \end{pmatrix} .
\\
&   \mathcal{R}^F(u,v)= \begin{pmatrix}
        1&\cdot&\cdot&f(u,v)\\
        \cdot&1&\cdot&\cdot\\
        \cdot&\cdot&1&\cdot\\
        \cdot&\cdot&\cdot&-1
    \end{pmatrix},\\
   &
\mathcal{R}^G(u,v)=\begin{pmatrix}
        f_1(u,v)&f_2(u,v)&f_3(u,v)&f_4(u,v)\\
        \cdot&f_1(u,v)&\cdot&f_3(u,v)\\
        \cdot&\cdot&f_1(u,v)&f_2(u,v)\\
        \cdot&\cdot&\cdot&f_1(u,v)
    \end{pmatrix},\\
& \mathcal{R}^{H}(u,v)=\begin{pmatrix}
        1&-p&p&c\;(f(u)-f(v))\\
        \cdot&1&\cdot& -q  \\
        \cdot&\cdot&1&q\\
        \cdot&\cdot&\cdot&1
    \end{pmatrix} .
    \end{align}

\paragraph{General 8-vertex type}
Then there are finally two additional 8-vertex type models that are of difference form
\begin{align}
    &
    \mathcal{R}^I(u,v)= \begin{pmatrix}
        1&\cdot&\cdot&k(u) e^{g(u)-g(v)}\\
        \cdot&-1&\cdot&\cdot\\
        \cdot&2 e^{g(u)-g(v)}&1&\cdot\\
        \cdot&\cdot&\cdot&1
    \end{pmatrix},\\
    &
    \mathcal{R}^J(u,v)= \begin{pmatrix}
        1&\cdot&\cdot&e^{g(u)-g(v)}\\
        \cdot&-1&e^{g(u)-g(v)}&\cdot\\
        \cdot&e^{g(u)-g(v)}&1&\cdot\\
        -e^{g(u)-g(v)}&\cdot&\cdot&1
    \end{pmatrix}.
      \label{R^J(u,v)}
\end{align}

\subsection{Rank 3}

We find several  solutions of rank 3
\begin{align}
   &  \mathcal{R}^K(u,v)=\begin{pmatrix}
       \cdot&p&  e^{g(u)-g(v)} p&f(u,v)\\
        \cdot&\cdot&e^{g(u)-g(v)}k&e^{g(u)-g(v)} q\\
        \cdot&k&\cdot&q\\
        \cdot&\cdot&\cdot&\cdot
    \end{pmatrix},\\
 &  \mathcal{R}^L(u,v)=\begin{pmatrix}
       1&\cdot&\cdot&\cdot\\
        \cdot&\cdot&\cdot&1\\
        \cdot&e^{f(u)-f(v)}&\cdot&1-e^{f(u)-f(v)}\\
        \cdot&\cdot&\cdot&1
    \end{pmatrix},
\end{align}

\begin{align}
   \mathcal{R}^M(u,v)&=\begin{pmatrix}
       e^{h(u)-h(v)}&\cdot&\cdot&\cdot\\
        \cdot&\cdot&e^{f(u)-f(v)}&\cdot\\
        \cdot&e^{g(u)-g(v)}&\cdot&\cdot\\
        \cdot&\cdot&\cdot&\cdot
    \end{pmatrix},\\
    \mathcal{R}^{M'}(u,v)&=\begin{pmatrix}
       e^{h(u)-h(v)}&\cdot&\cdot&\cdot\\
        \cdot&\cdot&\cdot&\cdot\\
        \cdot&e^{g(u)-g(v)}&\cdot&\cdot\\
        \cdot&\cdot&\cdot&e^{f(u)-f(v)}
    \end{pmatrix},
\\
   \mathcal{R}^N(u,v)&=\begin{pmatrix}
        1&\cdot&\cdot&f(u)h(v)\\
        \cdot&1&\cdot&-f(u)\\
       \cdot&\cdot&1&h(v)\\
        \cdot&\cdot&\cdot&\cdot
    \end{pmatrix},\\
  \mathcal{R}^{N'}(u,v)& =\begin{pmatrix}
        1&\cdot&\cdot&\cdot\\
        \cdot&1&\cdot&\cdot\\
        f(u)&-f(u)h(v)&\cdot&h(v)\\
        \cdot&\cdot&\cdot&1
    \end{pmatrix}.
    \end{align}

\subsection{Rank 2}
We find three solutions of rank 2
\begin{align}
    \mathcal{R}^O(u,v)&=\begin{pmatrix}
\cdot & \cdot &f_1(u,v) & f_2(u,v) \\
\cdot & \cdot & f_3(u,v) & f_4(u,v) \\
 \cdot & \cdot  & \cdot & \cdot \\
 \cdot & \cdot & \cdot & \cdot
    \end{pmatrix},\\
    \mathcal{R}^P(u,v)&=\begin{pmatrix}
       \cdot&f_1(u,v)&\cdot&f_2(u,v)\\
        \cdot&\cdot&\cdot&f_3(u,v)\\
        \cdot&\cdot&\cdot&\cdot\\
        \cdot&\cdot&\cdot&\cdot
    \end{pmatrix},\\
      \mathcal{R}^Q(u,v)&=\begin{pmatrix}
       \cdot&p-k&\frac{k-p}{k-q}(q+k)&f(u,v)\\
        \cdot&\cdot&\cdot&q+k\\
        \cdot&\cdot&\cdot&\frac{k+q}{k+p}(p-k)\\
        \cdot&\cdot&\cdot&\cdot
    \end{pmatrix}.
\end{align}

\subsection{Rank 1}
There are two solutions of rank 1
\begin{align}
    \mathcal{R}^R(u,v)&=\begin{pmatrix}
       \cdot&f_1(u,v)&f_2(u,v)&f_3(u,v)\\
        \cdot&\cdot&\cdot&\cdot\\
        \cdot&\cdot&\cdot&\cdot\\
        \cdot&\cdot&\cdot&\cdot
    \end{pmatrix},\\
    \mathcal{R}^S(u,v)&=\begin{pmatrix}
       \cdot&\cdot&\cdot&\cdot\\
        \cdot&f_1(u,v)&f_2(u,v)&\cdot\\
        \cdot&\cdot&\cdot&\cdot\\
        \cdot&\cdot&\cdot&\cdot
    \end{pmatrix}.
\end{align}
We checked that we reproduce all the $R$-matrices from \cite{Garkun:2024jnp} and we have completed their classification.

\section{R-matrices and Lax operators}
\label{R-matrices and Lax operators}

The relation between solutions of the Yang--Baxter equation and integrable spin chains follows the path of the algebraic Bethe Ansatz. The key idea is to introduce a Lax operator $L$ which will be used to define a transfer matrix that encodes the family of conserved charges that define the integrable model. The $R$-matrix then describes the algebra between the components of the Lax via the so-called $RLL$-relations or fundamental commutation relations \eqref{eq:FCR}. This is a well-known result, see for instance \cite{Faddeev:1987ih,Faddeev:1996iy}.

\begin{prop}\label{propRLL}
Consider a Lax operator $L_{an}(u) :V \otimes V  \rightarrow V \otimes V $ which satisfies the fundamental commutation relations
\begin{align}\label{eq:FCR}
R_{ab}(u,v) L_{an}(u) L_{bn}(v)  = L_{bn}(v)  L_{an}(u) R_{ab}(u,v),
\end{align}
for some invertible matrix $R: V\otimes V \rightarrow V \otimes V$. Let us furthermore assume that the entries of $L(u)$ are analytic functions of $u$. Then $L$ describes a spin chain with local Hilbert space $V$ with a tower of conserved charges.
\end{prop}
\begin{proof}
Construct the transfer matrix
\begin{align}\label{eq:transfer}
t(u) = \mathrm{tr}_a \big[ L_{aN}(u) \ldots L_{a1}(u)\big].
\end{align}
From the Fundamental Commutation Relations \eqref{eq:FCR}, you can show that
\begin{align}\label{eq:tt}
[t(u),t(v) ] =0
\end{align}
for all values of $u,v$. Since $L$ is analytic, we have the following expansion for the transfer matrix
\begin{align}
t(u) = \sum_{n=0}^\infty \tilde{Q}_n u^n.
\end{align}
Equation \eqref{eq:tt} then implies that $[\tilde{Q}_m,\tilde{Q}_n]=0$.
\end{proof}

In order for such a system to be integrable you also need to argue that they are independent. This will of course depend on the explicit form of the Lax operator, but generically that will be the case. Traditionally one defines the different conserved charges via the logarithmic derivatives of the transfer matrix \eqref{eq:transfer}. So, let us assume that $\log t$ is well-defined and admits a power series in $u$. In particular, let us define
\begin{align}
Q_{n+1} \equiv \frac{d^n}{du^n} \log t(u) \Big|_{u=0},
\end{align}
then from \eqref{eq:tt} it follows that
\begin{align}
[Q_m, Q_n] = 0.
\end{align}
When $L$ is a regular operator, \textit{i.e.} $L(0)=P$, then one can show that the operators $Q_n$ are local operators with interaction range at most $n$. In this case, the Hamiltonian is usually identified with $Q_2$. This is an operator with nearest-neighbour interactions, given by
\begin{align}
\mathcal{H} \equiv Q_2 = \sum_i P L_{i,i+1}^\prime (0).
\end{align}
In this case we see that the Hamiltonian can be identified with the first non-trivial term in the expansion of the Lax operator. For regular solutions there is a clear connection between the Lax operator and regular solutions of the Yang--Baxter equation. 

\begin{prop}\label{lem:RegularRLL}
Let $L$ be a regular Lax operator which satisfies the fundamental commutation relations. Then the corresponding $R$ matrix is regular as well
\begin{align}
R(u,u) = P.
\end{align}
\end{prop}
\begin{proof}
By evaluating the RLL relations at $u=v=0$, we find that
\begin{align}
R(0,0) \sim P.
\end{align}
Since $R$ is fixed up to an overall normalisation, we can set $R(0,0)=P$ without loss of generality. Now let us expand the Lax operator  and the $R$-matrix as
\begin{align}\label{eq:Lexpand}
&L(u) = P(1+ u \mathcal{H} + \sum_{i=2}^\infty u^i \check{L}^{(i)} ),
&& R(u,u) = P(1+ \sum^\infty_{i=1} u^i \check{R}^{(i)} ).
\end{align}
We plug this into the RLL relations and set $v=u$ to find
\begin{align}
&P_{ab}(1+ u \check{R}_{ab}^{(1)} + \ldots) P_{an}(1+u \mathcal{H}_{an} +\ldots) P_{bn}(1+u \mathcal{H}_{bn} +\ldots) =\nonumber  \\
 &\qquad \qquad P_{bn}(1+u \mathcal{H}_{bn} +\ldots) P_{an}(1+u \mathcal{H}_{an} +\ldots) P_{ab}(1+ u \check{R}_{ab}^{(1)} + \ldots).
\end{align}
At leading order this is satisfied by construction and at the next order we find
\begin{align}
&P_{ab}  \check{R}_{ab}^{(1)} P_{an} P_{bn} + P_{ab}  P_{an} \mathcal{H}_{an}P_{bn} + P_{ab} P_{an} P_{bn} \mathcal{H}_{bn}
=\nonumber  \\
 &\qquad \qquad 
 P_{bn}P_{an}  P_{ab}  \check{R}_{ab}^{(1) }+  P_{bn}  P_{an} \mathcal{H}_{an}P_{ab} +  P_{bn} \mathcal{H}_{bn}P_{an}P_{ab} ,
\end{align}
from which it follows that $\check{R}_{ab}^{(1)} =\check{R}_{bn}^{(1)}$ which can only be true for $\check{R}_{ab}^{(1)} \sim 1$. Hence this coefficient can be put to 0 by choosing again an appropriate normalisation of $R$. One can then recursively show that all orders in $R$ can be chosen to vanish. More specifically, by cancelling the permutation operators in the RLL relations we find
\begin{align}
&(1+ \sum^\infty_{i=1} u^i \check{R}^{(i)}_{bn} )(1+u \mathcal{H}_{ab} +\ldots)(1+u \mathcal{H}_{bn} +\ldots) 
=\nonumber\\
&\qquad\qquad (1+u \mathcal{H}_{ab} +\ldots) (1+u \mathcal{H}_{bn} +\ldots) (1+ \sum^\infty_{i=1} u^i \check{R}^{(i)}_{ab} ).
\end{align}
This yields an equation for $\check{R}^{(i)}$ in terms of the lower order expansion terms of the form
\begin{align}
\check{R}^{(i)}_{bn} + \sum_{k+l+m=i} \check{R}^{(k)}_{bn}\check{L}^{(l)}_{ab}\check{L}^{(m)}_{bn} = \check{R}^{(i)}_{ab} + \sum_{k+l+m=i} \check{L}^{(l)}_{ab}\check{L}^{(m)}_{bn} \check{R}^{(k)}_{ab},
\end{align}
where the indices in the sum run over positive integers, \textit{i.e.} $l,m>0$. By the induction hypothesis we have that $\check{R}^{(k)}=0$ for $k<i$, from which it follows that $\check{R}^{(i)}_{bn} = \check{R}^{(i)}_{ab}$. This again implies that $\check{R}^{(i)}\sim1$ and can be set to 0 by an appropriate normalisation.
\end{proof}

We can now show that there is a one-to-one correspondence between regular Lax operators of integrable spin chains and regular solutions of the Yang--Baxter equation.

\begin{thm}
\label{thm}
Any $R$-matrix which satisfies the RLL relations for a regular Lax operator is a solution of the quantum Yang-Baxter equation
\begin{align}
R_{12}(u_1,u_2) R_{13}(u_1,u_3) R_{23}(u_2,u_3) = R_{23}(u_2,u_3) R_{13} (u_1,u_3)R_{12}(u_1,u_2),
\end{align}
and satisfies the braiding unitarity relation
\begin{align}\label{eq:Braiding}
R_{21}(v,u) R_{12}(u,v) = 1. 
\end{align}
Conversely, any regular solution $R$ of the quantum Yang--Baxter equation can be interpreted as a regular Lax matrix. More specifically, there is a correspondence between regular solutions of the Yang--Baxter equations and regular Lax operators defined by $L(u) \simeq R(u,0)$.
\end{thm}
\begin{proof}
The last part of the theorem is straightforward to prove. From the Yang--Baxter equation it directly follows that $L(u) \simeq R(u,0)$ satisfies the fundamental commutation relations \eqref{eq:FCR}. Conversely, consider the fundamental commutation relations and set $v=0$. This gives
\begin{align}
R_{an}(u,0) L_{ab}(u)  =   L_{an}(u) R_{ab}(u,0).
\end{align}
Since $L$ is regular, it is invertible and we have that
\begin{align}
 L^{-1}_{an}(u) R_{an}(u,0)  =   R_{ab}(u,0)L^{-1}_{ab}(u) .
\end{align}
Both sides have a different index structure and act on different spaces. The only solution to this equation is that $ L^{-1}_{an}(u) R_{an}(u,0) \sim 1$, which completes the first part of the proof.
\medskip

Let us now prove the first part of the theorem. Let $L(u)$ be an invertible Lax operator that satisfies the RLL relations \eqref{eq:FCR}. From the fundamental commutation relations we find that
\begin{align}\label{eq:braid}
[R_{21}(v,u) R_{12} (u,v), L_{1n} (u)L_{2n}(v) ] = 0.
\end{align}
Let us define the object $R_{21} R_{12}\equiv X_{12}$. From Proposition \ref{lem:RegularRLL} and our assumption that the $R$-matrix is analytic, we see that $X$ admits a power series of the form
\begin{align}
X(u,v) = 1 + \sum_{i=1}^\infty u_-^i X^{(i)}(u_+),
\end{align}
where $u_\pm$ are defined as in Proposition \ref{prop1}. Let us define $\check{L} \equiv PL$. We can then, analogously to the proof of the previous lemma, expand this around the point $u=v$, or equivalently $u_-=0$. At leading order this equation is trivially satisfied. At the next order we find
\begin{align}
X^{(1)}_{2n}(u) \check{L}_{12}(u)\check{L}_{2n}(u)   =  \check{L}_{12}(u)\check{L}_{2n}(u) X^{(1)}_{12} (u).
\end{align}
Using the expansion of the Lax operator \eqref{eq:Lexpand}, we then find
\begin{align}
&X^{(1)}_{2n}(u)(1+u \mathcal{H}_{12} +\ldots) (1+u \mathcal{H}_{2n} +\ldots)  =  (1+u \mathcal{H}_{12} +\ldots) (1+u \mathcal{H}_{2n} +\ldots) X^{(1)}_{12}(u).
\end{align}
Expanding this as a power series in $u$ proves that $X^{(1)} \sim 1$. From this one can easily show by induction that $X\sim1$ and hence that the $R$-matrix satisfies braiding unitarity. Moreover, it is easy to show that the proportionality factor needs to be symmetric under changing $u,v$ and because of this it can be absorbed into the prefactor of the $R$-matrix, making braiding unitarity exact.

\medskip

What remains to show is that $R$ satisfies the quantum Yang--Baxter equation. To show this, we need to consider a triple product of Lax operators
\begin{align}
L_{1n}(u)L_{2n}(v)L_{3n}(w).
\end{align}
Then you can show that 
\begin{align}\label{eq:RRRRRR}
[R_{32}R_{31}R_{21}R_{23}R_{13}R_{12} , L_{1n}L_{2n}L_{3n}]=0,
\end{align}
where we have suppressed the explicit spectral parameter dependence. This is illustrated in Figure \ref{FigYBE}. One can use the $R$-matrices in the above chain to permute the Lax operators. The chain is chosen such that it leaves $L_1L_2L_3$ invariant.
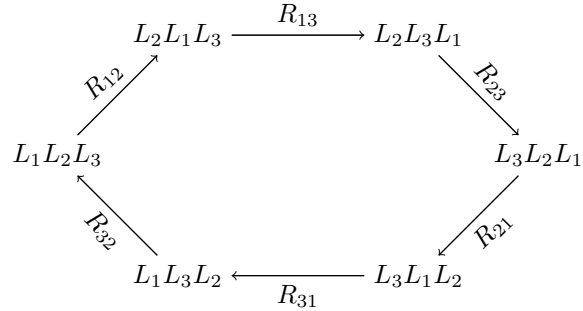
\begin{figure}[h!]
\begin{center}
\begin{tikzpicture}[scale=.8]
\node at (0,0) (nodeA) {$L_1 L_2 L_3$};
\node at (2,2) (nodeB) {$L_2 L_1 L_3$};
\node at (6,2) (nodeC) {$L_2 L_3 L_1$};
\node at (8,0) (nodeD) {$L_3 L_2 L_1$};
\node at (2,-2) (nodeE) {$L_1 L_3 L_2$};
\node at (6,-2) (nodeF) {$L_3 L_1 L_2$};
\draw[->] (nodeA) -- (nodeB);
\draw[->] (nodeB) -- (nodeC);
\draw[->] (nodeC) -- (nodeD);
\draw[->] (nodeE) -- (nodeA);
\draw[->] (nodeF) -- (nodeE);
\draw[->] (nodeD) -- (nodeF);
\draw (nodeA) -- (nodeB) node [midway, above, sloped] (TextNode) {$R_{12}$};
\draw (nodeB) -- (nodeC) node [midway, above, sloped] (TextNode) {$R_{13}$};
\draw (nodeC) -- (nodeD) node [midway, above, sloped] (TextNode) {$R_{23}$};
\draw (nodeD) -- (nodeF) node [midway, below, sloped] (TextNode) {$R_{21}$};
\draw (nodeF) -- (nodeE) node [midway, below, sloped] (TextNode) {$R_{31}$};
\draw (nodeA) -- (nodeE) node [midway, below, sloped] (TextNode) {$R_{32}$};
\end{tikzpicture}
\end{center}
\caption{Acting with different $R$-matrices will permute the Lax operators in the triple product $L_1L_2L_3$ by using the fundamental commutation relations. If we follow the arrows in the diagram, we see that $R_{32}R_{31}R_{21}R_{23}R_{13}R_{12}$ commutes with $L_1L_2L_3$. }\label{FigYBE}
\end{figure}

\noindent For brevity, let us write
\begin{align}
A_{123} \equiv R_{32}R_{31}R_{21}R_{23}R_{13}R_{12}.
\end{align}
We have already shown that for $w=0$ the $R$-matrix reduces to the Lax operator. We assume that the fundamental commutation relations hold, which this means that the YBE is satisfied in this case
\begin{align}
R_{12}(u,v) R_{13}(u,0) R_{23}(v,0) = R_{23}(v,0) R_{13} (u,0)R_{12}(u,v).
\end{align}
Let us make use of this fact by expanding around the case where one of the spectral parameters vanishes. In particular, consider
\begin{align}
A_{123}(u,v,w)L_{1n}(u)L_{2n}(v)L_{3n}(w)=L_{1n}(u)L_{2n}(v)L_{3n}(w)A_{123}(u,v,w).
\end{align}
We set $u=0$ and use regularity to get
\begin{align}
A_{123}(0,v,w)P_{1n}L_{2n}(v)L_{3n}(w)=P_{1n}L_{2n}(v)L_{3n}(w)A_{123}(0,v,w).
\end{align}
Cancelling the permutation operators then gives
\begin{align}
A_{n23}(0,v,w)L_{2n}(v)L_{3n}(w)=L_{2n}(v)L_{3n}(w)A_{123}(0,v,w).
\end{align}
We see that the left hand side does not depend on site 1 anymore. Hence $A_{123}(0,v,w) = \tilde{A}_{23}(v,w)$ and we obtain
\begin{align}
\tilde{A}_{23}(v,w)L_{2n}(v)L_{3n}(w)=L_{2n}(v)L_{3n}(w)\tilde{A}_{23}(v,w).
\end{align}
But now we recognise the same equations as in the case of braiding unitarity \eqref{eq:braid} and we can similarly conclude that $A_{123}(0,v,w) \sim 1$. Let us now expand to the next order in $u$. We get
\begin{align}
&A_{123}(0,v,w)L^\prime_{1n}(0)L_{2n}(v)L_{3n}(w) + A^\prime_{123}(0,v,w)P_{1n}L_{2n}(v)L_{3n}(w) = \\
&\qquad L^\prime_{1n}(0)L_{2n}(v)L_{3n}(w)A_{123}(0,v,w) + P_{1n}L_{2n}(v)L_{3n}(w)A^\prime_{123}(0,v,w).
\end{align}
Using  $A_{123}(0,v,w) \sim 1$ we derive
\begin{align}
 A^\prime_{123}(0,v,w) P_{1n}L_{2n}(v)L_{3n}(w) =  P_{1n}L_{2n}(v)L_{3n}(w)A^\prime_{123}(0,v,w),
\end{align}
which is exactly the same equation that we just solved. Hence, by induction we find that  $A_{123}(u,v,w) = f(u,v,w) 1 $, where $f$ is an analytic function.

Concluding, we showed that
\begin{align}
\label{YBEsim}
R_{32}(w,v)R_{31}(w,u)R_{21}(v,u)R_{23}(v,w)R_{13}(u,w)R_{12}(u,v) =f(u,v,w).
\end{align}
By braiding unitarity we can derive
\begin{align}\label{eq:RRR1}
R_{23}(v,w)R_{13}(u,w)R_{12}(u,v) =f(u,v,w) R_{12}(u,v)R_{13}(u,w)R_{23}(v,w).
\end{align}
This is not quite the Yang Baxter equation. We still need to show $f(u,v,w)=1$. Using Braiding unitarity again \eqref{YBEsim} also leads to the following rewriting
\begin{align}
R_{32}(w,v)R_{31}(w,u)R_{21}(v,u) &=f(u,v,w) R_{21}(v,u)R_{31}(w,u)R_{32}(w,v),\\
R_{32}(w,v)P_{13}R_{13}(w,u)P_{13}R_{21}(v,u)& =f(u,v,w) R_{21}(v,u)P_{13}R_{13}(w,u)P_{13}R_{32}(w,v),\\
R_{12}(w,v)R_{13}(w,u)R_{23}(v,u) &=f(u,v,w) R_{23}(v,u)R_{13}(w,u)R_{12}(w,v).
\end{align}
We can now relabel the parameters $u\rightarrow w$ and $w\rightarrow u$ and obtain
\begin{align}
R_{12}(u,v)R_{13}(u,w)R_{23}(v,w) &=f(w,v,u) R_{23}(v,w)R_{13}(u,w)R_{12}(u,v).
\end{align}
By comparing this against \eqref{eq:RRR1} we conclude that
\begin{align}
f(u,v,w)=\frac{1}{f(w,v,u)}.
\end{align}
Therefore $f$ has to be of the form
\begin{align}
\label{fform}
f(u,v,w)=\pm e^{A(v,w)-A(u,w)}.
\end{align}
By our identification of $R$ with the Lax operator, we know that putting any of the parameters to zero leads to $f(0,v,w)=f(u,0,w)=f(u,v,0)=1$. Together with \eqref{fform}, this leads to $A=const$ and therefore $f=const=1$. Which proves that any matrix $R$, satisfying the fundamental commutation relations for a regular Lax operator, is a regular solution of the Yang Baxter Equation.

\end{proof}

\section{Non-regular Lax operators and the modified YBE}
\label{Non-regular Lax operators and the modified YBE}

For non-regular solutions of the Yang-Baxter equation the correspondence between Lax operators and $R$-matrices is lost. As we saw in Proposition \ref{propRLL}, the defining relation that implies a tower of conserved charges and integrability are the fundamental commutation relations \eqref{eq:FCR}. It is clear that we can use every non-regular solution of the Yang--Baxter equation that we found as such a Lax operator $L(u) \equiv \mathcal{R}(u,0)$. It is clear that the corresponding $R$-matrix from the fundamental commutation relations exists. However, unlike in the regular case, there might be a more general solution. In order to show this, we now consider
\begin{align}
    \tilde{R}_{12}(u,v) L_{13}(u) L_{23}(v)=L_{23}(v)L_{13}(u)  \tilde{R}_{12}(u,v), 
\end{align}
with $L$ a non-regular Lax matrix as defined above. We can then simply solve for $\tilde{R}$ and see what we find. 

For some solutions, there is a more general $\tilde{R}$-matrix than the original $R$-matrix. In particular, we find that all of the $\tilde{R}$ matrices listed below can be chosen to have a regular limit.  Furthermore, unlike in the regular case, $\tilde{R}$ \textit{does not solve the $YBE$}, but solves a modified Yang-Baxter equation instead
\begin{align}
\Tilde{R}_{12}(u,v)\Tilde{R}_{13}(u,w)\Tilde{R}_{23}(v,w) = M_{123}(u,v,w) \Tilde{R}_{23}(v,w) \Tilde{R}_{13}(u,w) R_{12}(u,v),
\end{align}
or, alternatively,
\begin{align}
\Tilde{R}_{12}(u,v)\Tilde{R}_{13}(u,w)\Tilde{R}_{23}(v,w) - \Tilde{R}_{23}(v,w) \Tilde{R}_{13}(u,w) R_{12}(u,v)=\tilde{M}_{123}(u,v,w).
\end{align}
The new matrix $\tilde{R}$-matrix does satisfy equation \eqref{eq:RRRRRR}, but since the Lax operator is not regular anymore there are more solutions to this equation that just the identity operator. This is the crucial point where the regular and non-regular solutions differ and where the proof of the previous section breaks down.

There will be instances when $\tilde{R}$ satisfies the usual Yang--Baxter equation and is not given by the non-regular $R$-matrix that we started out from. In these cases we can interpret that $R$-matrix as a Lax operator of a different spin chain. In particular, there will instances where the $R$-matrix can for instance be seen as a non-regular Lax operator of the $R$-matrix of the XXX spin chain. This is useful in solving these models via the Algebraic Bethe Ansatz since the known fundamental commutations can be used and the only thing that changes is the action of the monodromy matrix on the reference state.

Certain deformations of the YBE are well-known. For example in \cite{gerstenhaber1997boundarysolutionsquantumyangbaxter} a different modified YBE has been proposed. Of course there is also the generalised or twisted YBE \cite{KASHAEV_1993} \cite{tarasov1994solutionstwistedyangbaxterequation} and other proposed models such as the dynamical YBE \cite{Etingof_1998}. So far we have not seen a clear correspondence from any of these equations to ours.

\subsection{Example via Braiding unitarity}

An alternative way to see this is to note that a lot of these $R$-matrices do not satisfy braiding unitarity. Because of this, we find that there are two possible solutions of the fundamental commutation relations, namely $R(u,v)$ and $(PR(v,u)P)^{-1} \equiv \hat{R}$. For solutions that satisfy braiding unitarity, these obviously coincide. If they do not, then any linear combination of these will be a solution to the fundamental commutation relations and while both $R(u,v)$ and $(PR(v,u)P)^{-1}$ satisfy the Yang--Baxter equation separately, their sum will generically not. A perfect example of this is $\mathcal{R}^G$. It is easy to show that this model does not satisfy braiding unitarity and we find
\begin{align}
\hat{\mathcal{R}}^G \sim 
\begin{pmatrix}
 1 & 0 & 0 & -e^{v-u} \\
 0 & 1 & e^{v-u} & 0 \\
 0 & e^{v-u} & -1 & 0 \\
 e^{v-u} & 0 & 0 & 1 \\
\end{pmatrix}.
\end{align}
Now, we can show by direct computation that the most general solution of the fundamental commutation relations when $L(u) = \mathcal{R}^G(u,0)$ is given by
\begin{align}\label{eq:RtildeG}
\tilde{\mathcal{R}}^G = f(u,v) \mathcal{R}^G  + g(u,v) \hat{\mathcal{R}}^G .
\end{align}
It is also easy to check that $\tilde{\mathcal{R}}^G$ does not satisfy the Yang--Baxter equation unless $fg=0$ or $f=g$. Note that in the latter case $\tilde{\mathcal{R}}^G$ is a regular solution of the Yang--Baxter equation and is of the form 8 vertex A from \cite{deLeeuw:2020ahe}. In other words, we find that $\mathcal{R}^G$ can be interpreted as a non-regular Lax operator for the $R$-matrix of a free fermion model.

\subsection{Other example}

Apart from \eqref{eq:RtildeG} we found several more cases of models with a non-trivial $\tilde{\mathcal{R}}$. In some cases we find that $\tilde{\mathcal{R}}$ can be chosen to be regular and satisfying the Yang--Baxter equation. In this case, the non-regular $R$-matrix can be interpreted as a non-regular Lax operator corresponding satisfying the fundamental commutation relations with a known $R$-matrix. 

\paragraph{Model $\mathcal{R}^C$}
\begin{align}
    \tilde{\mathcal{R}}^C(u,v)=\begin{pmatrix}
        g_1(u,v)&\cdot&\cdot&g_3(u,v)\\
        \cdot&g_2(u,v)& \frac{f(u,0)}{f(v,0)}(g_1(u,v)+g_2(u,v))&\cdot\\
        \cdot&\frac{f(v,0)}{f(u,0)}(g_1(u,v)+g_2(u,v))&g_2(u,v)&\cdot\\
        \cdot&\cdot&\cdot&g_1(u,v)
    \end{pmatrix},
\end{align}
for free functions $g_1,g_2$. We see that this $ \tilde{\mathcal{R}}$ reduces to the permutation operator at $u=v$ if we choose $g_1$ and $g_2$ such that
\begin{align}
&g_1(u,u) = 1, && g_2(u,u) = 0, && g_3(u,u) = 0.
\end{align}
This solution $\tilde{\mathcal{R}}$ is regular and satisfies the Yang--Baxter equation exactly when
\begin{align}
&g_1(u,v) = g(u,v),
&&g_2(u,v) = \frac{g(u,v)}{G(u)-G(v)-1},
\end{align}
where $g$ is an free function and $G'=g$. If we set $g_3=0$, then this is simply the $R$-matrix of the twisted XXX spin chain up to a local basis transformation.

\section{Conclusions and outlook}

In this paper we completed the classification of $4\times 4$ solutions of the quantum Yang--Baxter equation by considering non-regular $R$-matrices. We apply a perturbative approach where we expand about the non-regular solution $R(u,u)$ which were classified previously \cite{HIETARINTA}.  We reproduce the results from \cite{Garkun:2024jnp} and extend their results to solutions that are of non-difference form and of lower rank.

We furthermore show that there is a one-to-one relation between regular solutions of the Yang--Baxter equation and regular $R$-matrices. However, for non-regular models this identification is lost. We find for instance that some non-regular $R$-matrices can be interpreted as non-regular Lax operator that satisfy the fundamental commutation relations for a regular $R$-matrix. Generically the $R$-matrices that describe the fundamental commutation relations in these cases satisfy a modified Yang--Baxter equation.

It would be interesting to further understand the modified Yang--Baxter equation and to find if there is a universal form to the modification. So far we have not been able to identify one. Furthermore, extending our results to higher dimensions is another possible avenue of research. There is also the possibility of higher dimensional spectral parameters \cite{Korepin_Bogoliubov_Izergin_1993}, which could be explored. Maybe the approach of \cite{Maity:2024vtr} can used to understand certain classes of non-regular solutions. Finally, it could also be interesting to find the physical properties of the models corresponding to our solutions. For instance the spin chains or the quantum circuits that are associated to them.

\paragraph{Acknowledgements.}

We are grateful to S. Shatashvili and L. Takhtajan for important remarks on the introduction. We would like to thank L. Corcoran, V. Grivtsev and A.L. Retore for useful discussions. MdL was supported in part by SFI and the Royal Society for funding under grants UF160578, RGF$\backslash$ R1$\backslash$ 181011, RGF$\backslash$8EA$\backslash$180167 and RF$\backslash$ ERE$\backslash$ 210373. MdL is also supported by ERC-2022-CoG - FAIM 101088193. VP was supported by ERC-2022-CoG - FAIM 101088193.

\appendix

\section{Constant solutions}
\label{constantsolutions}
The starting point for our analysis are the constant solutions of the Yang--Baxter equation. They were classified in \cite{HIETARINTA} and we list them here. We omit the permutation operator, which gives rise to regular solutions who were classified already in \cite{Corcoran:2023zax}.  We group the matrices according to their rank. 

\subsection{Rank 4} 

Here we list the invertible $R$-matrices. These were considered recently in \cite{Maity:2024vtr} as well.

\begin{align}
   R^{A}& = \begin{pmatrix}
        k&\cdot&\cdot&\cdot\\
        \cdot&q&\cdot&\cdot\\
        \cdot&\cdot&p&\cdot\\
        \cdot&\cdot&\cdot&s
    \end{pmatrix},
&&   R^B= \begin{pmatrix}
        k&q&p&s\\
        \cdot&k&\cdot&p\\
        \cdot&\cdot&k&q\\
        \cdot&\cdot&\cdot&k
    \end{pmatrix} ,\\
        R^C&=\begin{pmatrix}
        k^2&\cdot&\cdot&\cdot\\
        \cdot&q k&\cdot&\cdot\\
        \cdot&k^2-qp&p k&\cdot\\
        \cdot&\cdot&\cdot&k^2
    \end{pmatrix},
    &&   R^{C'}=\begin{pmatrix}
        k^2&\cdot&\cdot&\cdot\\
        \cdot&q k&\cdot&\cdot\\
        \cdot&k^2-qp&p k&\cdot\\
        \cdot&\cdot&\cdot&-pq
    \end{pmatrix}, \\
   R^D&=\begin{pmatrix}
        \cdot&\cdot&\cdot&q\\
        \cdot&\cdot&k&\cdot\\
        \cdot&k&\cdot&\cdot\\
        p&\cdot&\cdot&\cdot
    \end{pmatrix},
&&   R^E = \begin{pmatrix}
        k^2&-kp&kp&pq\\
        \cdot&k^2&\cdot&-kq\\
        \cdot&\cdot&k^2&kq\\
        \cdot&\cdot&\cdot&k^2
    \end{pmatrix},
\\
   R^F&= \begin{pmatrix}
        p&\cdot&\cdot&k\\
        \cdot&q&\cdot&\cdot\\
        \cdot&p-q&p&\cdot\\
        \cdot&\cdot&\cdot&-q
    \end{pmatrix},
&&
  R^G=\begin{pmatrix}
        1+2q-q^2&\cdot&\cdot&1-q^2\\
        \cdot&1+q^2&1-q^2&\cdot\\
        \cdot&1-q^2&1+q^2&\cdot\\
        1-q^2&\cdot &\cdot&1-2q-q^2
    \end{pmatrix},
\\
   R^H&= \begin{pmatrix}
        1&\cdot&\cdot&1\\
        \cdot&-1&1&\cdot\\
        \cdot&1&1&\cdot\\
        -1&\cdot&\cdot&1
    \end{pmatrix},
&&
   R^I= \begin{pmatrix}
        1&\cdot&\cdot&1\\
        \cdot&-1&\cdot&\cdot\\
        \cdot&\cdot&-1&\cdot\\
        \cdot&\cdot&\cdot&1
    \end{pmatrix}.
\end{align}

\subsection{Rank 3} 
\begin{align}
   R^J&=\begin{pmatrix}
        p+q&\cdot&\cdot&\cdot\\
        \cdot&q&\cdot&q\\
        \cdot&\cdot &p+q&\cdot\\
        \cdot&p &\cdot&p
    \end{pmatrix},
&&   R^K = \begin{pmatrix}
        \cdot&p&p&\cdot \\
        \cdot&\cdot&k&q\\
        \cdot&k&\cdot&q\\
        \cdot&\cdot&\cdot&\cdot
    \end{pmatrix},
\\    
   R^L&= \begin{pmatrix}
        1&\cdot&\cdot&\cdot\\
        \cdot&\cdot&\cdot&1\\
        \cdot&1&\cdot&\cdot\\
        \cdot&\cdot&\cdot&1
    \end{pmatrix},
&&     R^M= \begin{pmatrix}
        1&\cdot&\cdot&\cdot\\
        \cdot&\cdot &1&\cdot\\
        \cdot&1&\cdot &\cdot\\
        \cdot&\cdot&\cdot&\cdot
    \end{pmatrix}.
\end{align}

\subsection{Rank 2} 

\begin{align}
   R^N & =\begin{pmatrix}
        \cdot&(q-k)(p^2-k^2)&(q+k)(p^2-k^2)&s\\
        \cdot&\cdot&\cdot&(q^2-k^2)(p+k)\\
        \cdot&\cdot &\cdot&(q^2-k^2)(p-k)\\
        \cdot&\cdot &\cdot&\cdot
    \end{pmatrix},
&& R^O= \begin{pmatrix}
        1&1&1&\cdot \\
        \cdot&\cdot&\cdot&\cdot\\
        \cdot&\cdot&\cdot&\cdot\\
        \cdot&\cdot&\cdot&1
        \end{pmatrix}
\\
   R^P&=\begin{pmatrix}
        \cdot&p&\cdot&q\\
        \cdot&\cdot&\cdot&\cdot\\
        \cdot&k &\cdot&\cdot\\
        \cdot&\cdot &\cdot&\cdot
    \end{pmatrix}
&&
   R^Q= \begin{pmatrix}
        \cdot&p&\cdot&\cdot \\
        \cdot&\cdot&\cdot&q\\
        \cdot&\cdot&\cdot&\cdot\\
        \cdot&\cdot&\cdot&\cdot
    \end{pmatrix}
\\
   R^R&=\begin{pmatrix}
        \cdot&p&\cdot&\cdot\\
        \cdot&\cdot&\cdot&\cdot\\
        \cdot&\cdot &\cdot&q\\
        \cdot&\cdot &\cdot&\cdot
    \end{pmatrix}.
\end{align}

\subsection{Rank 1} 

\begin{align}
  R^S= \begin{pmatrix}
        \cdot&p&q&\cdot \\
        \cdot&\cdot&\cdot&\cdot\\
        \cdot&\cdot&\cdot&\cdot\\
        \cdot&\cdot&\cdot&\cdot
    \end{pmatrix}
    \quad
     R^T= \begin{pmatrix}
        \cdot&\cdot&\cdot&\cdot \\
        \cdot&p&q&\cdot\\
        \cdot&\cdot&\cdot&\cdot\\
        \cdot&\cdot&\cdot&\cdot
    \end{pmatrix} 
\end{align}

\section{ From constant $R$ to $\mathcal{R}(u,v)$}\label{completeR}
In this appendix one finds a detailed list of the non-constant models \ref{Non-constant Models} found from the constant ones \ref{constantsolutions}. Every section has a diagrammatic summary of these relations.
\subsection{$RA$}
Let's start with the diagonal model.
\label{sec:RA}
\begin{equation}
\label{RAconst}
   RA= \begin{pmatrix}
        k&\cdot&\cdot&\cdot\\
        \cdot&q&\cdot&\cdot\\
        \cdot&\cdot&p&\cdot\\
        \cdot&\cdot&\cdot&s
    \end{pmatrix}.
\end{equation}
\paragraph{\textcolor{BrickRed}{Non-constant Models}}
The non-constant models that can be found from \eqref{RAconst} are the following:
\begin{center}
\begin{tikzpicture}
 \draw[BrickRed!90, thick, ->] (0,0) -- (4,0) ;
  \draw[BrickRed!90, thick, ->] (0,0) -- (4,-1) ;
   \draw[BrickRed!90, thick, ->] (0,0) -- (4,1) ;
    \draw[BrickRed!90, thick, ->] (0,0) -- (4,-2) ;
     \draw[BrickRed!90, thick, ->] (0,0) -- (4,2) ;
   \filldraw[color=BrickRed!90, fill=BrickRed!10,  thick ](0,0) circle (0.5);
   \draw(4.5,0) circle (0.5);
    \draw (0,0) node {${RA}$};
    \draw (4.5,0) node {$\mathcal{R}^D$} ;  
   \draw(4.5,-1) circle (0.5);
    \draw (4.5,-1) node {$\mathcal{R}^F$} ;
    \draw(4.5,1) circle (0.5);
    \draw (4.5,1) node {$\mathcal{R}^C$} ;
    \draw(4.5,-2) circle (0.5);
    \draw (4.5,-2) node {$\mathcal{R}^S$} ;
    \draw(4.5,2) circle (0.5);
    \draw (4.5,2) node {$\mathcal{R}^A$} ;
\end{tikzpicture}
\end{center}
\paragraph{$\mathcal{R}^A$}
\begin{equation}
    \label{RA}
\mathcal{R}^A(u,v)=\begin{pmatrix}
        f_1(u,v)&\cdot&\cdot&\cdot\\
        \cdot&f_2(u,v)&\cdot&\cdot\\
        \cdot&\cdot&f_3(u,v)&\cdot\\
        \cdot&\cdot&\cdot&f_4(u,v)
    \end{pmatrix}.
\end{equation}
The diagonal non constant model can be found for general coefficients $k,q,p,s$.

\paragraph{$\mathcal{R}^C$}
\begin{align}
    \mathcal{R}^{C,i,j}(u,v)= f(u,v) \sigma^i \otimes \sigma^i + g(u,v) \sigma^j \otimes \sigma^j,
\end{align}

For $p=q=-1$, $k=s=1$ one finds $\mathcal{R}^C$.

\paragraph{$\mathcal{R}^D$}
\begin{align}
\label{RD}
    \mathcal{R}^D:=f(u,v) \sigma_z \otimes \sigma_z +g(u,v) \sigma_{\pm} \otimes \sigma_{\pm}.
\end{align}
Found at $RA(k=s=1,p=q=-1)$. 
\paragraph{$\mathcal{R}^F$}
\begin{equation}
\label{RF}
   \mathcal{R}^F:= \begin{pmatrix}
        1&\cdot&\cdot&f(u,v)\\
        \cdot&1&\cdot&\cdot\\
        \cdot&\cdot&1&\cdot\\
        \cdot&\cdot&\cdot&-1
    \end{pmatrix} .
\end{equation}
Found at $RA(k=p=q=1,s=-1)$.
\paragraph{$\mathcal{R}^S$}
\begin{equation}
\label{RS}
   \mathcal{R}^S:= \begin{pmatrix}
        \cdot&\cdot&\cdot&\cdot\\
        \cdot&f(u,v)&g(u,v)&\cdot\\
        \cdot&\cdot&\cdot&\cdot\\
        \cdot&\cdot&\cdot&\cdot
    \end{pmatrix} .
\end{equation}
Found at $RA(k=s=q=0,p=1)$.

\subsection{$RB$}
\label{sec:RB}
\begin{equation}
   RB= \begin{pmatrix}
        k&q&p&s\\
        \cdot&k&\cdot&p\\
        \cdot&\cdot&k&q\\
        \cdot&\cdot&\cdot&k
    \end{pmatrix}.
\end{equation}
\paragraph{Constant Model Intersections}
For $q=p=s=0$ we find a special case of $RA$ \ref{sec:RA}, this setting is not discussed again.
\paragraph{\textcolor{red}{Non-constant Models}}
\begin{center}
\begin{tikzpicture}
 \draw[red!100, thick, ->] (0,0) -- (4,0) ;
  \draw[red!100, thick, ->] (0,0) -- (4,-1) ;
  \draw[red!100, thick, ->] (0,0) -- (4,1) ;
  \draw[red!100, thick, ->] (0,0) -- (4,-2) ;
   \draw[red!100, thick, ->] (0,0) -- (4,2) ;
   \draw[red!100, thick, ->] (0,0) -- (4,3) ;
   \filldraw[color=red!100, fill=red!20,  thick ](0,0) circle (0.5);
  \draw (0,0) node {${RB}$};
     \draw(4.5,0) circle (0.5);
    \draw (4.5,0) node {$\mathcal{R}^G$} ;  
      \draw(4.5,-2) circle (0.5);
    \draw (4.5,-2) node {$\mathcal{R}^R$} ;    
      \draw(4.5,-1) circle (0.5);
    \draw (4.5,-1) node {$\mathcal{R}^O$} ;  
      \draw(4.5,2) circle (0.5);
    \draw (4.5,2) node {$\mathcal{R}^{D}$} ;  

  \draw(4.5,1) circle (0.5);
    \draw (4.5,1) node {$\mathcal{R}^F$} ; 
     \draw(4.5,3) circle (0.5);
    \draw (4.5,3) node {$\mathcal{R}^B$} ; 
\end{tikzpicture}
\end{center}
\paragraph{$\mathcal{R}^B$}
\begin{equation}
\label{RB}
    \mathcal{R}^B(u,v):=\begin{pmatrix}
        \cdot&\cdot&\cdot& a f_1(u,v)\\
        \cdot&\cdot& f_2(u,v)&\cdot\\
        \cdot& f_2(u,v)&\cdot&\cdot\\
        b f_1(u,v)&\cdot&\cdot&\cdot
    \end{pmatrix}.
\end{equation}
One finds \eqref{RB} for $k=p=q=0$.
\paragraph{$\mathcal{R}^{F}$}
One finds a model similar to \eqref{RF} for $k=p=q=0$.

\paragraph{$\mathcal{R}^G$}
\begin{equation}
    \label{RG}\mathcal{R}^G(u,v)=\begin{pmatrix}
        f_1(u,v)&f_2(u,v)&f_3(u,v)&f_4(u,v)\\
        \cdot&f_1(u,v)&\cdot&f_3(u,v)\\
        \cdot&\cdot&f_1(u,v)&f_2(u,v)\\
        \cdot&\cdot&\cdot&f_1(u,v)
    \end{pmatrix}.
\end{equation}
Found for all equal coefficients.
\paragraph{$\mathcal{R}^D$}
One finds \eqref{RD} for $k=s=0$ and $p,q$ generic.
\paragraph{$\mathcal{R}^O$}
\begin{equation}
\label{RO}
   \mathcal{R}^O:= \begin{pmatrix}
        \cdot&\cdot&f_1(u,v)&f_2(u,v)\\
        \cdot&\cdot&f_3(u,v)&f_4(u,v)\\
        \cdot&\cdot&\cdot&\cdot\\
        \cdot&\cdot&\cdot&\cdot
    \end{pmatrix} .
    \end{equation}
Found for $(k=s=q=0)$.

\paragraph{$\mathcal{R}^R$}
\begin{equation}
\label{RR}
   \mathcal{R}^{R}:= \begin{pmatrix}
        \cdot&\cdot&\cdot&f(u,v)\\
        \cdot&\cdot&\cdot&g(u,v)\\
        \cdot&\cdot&\cdot&h(u,v)\\
        \cdot&\cdot&\cdot&\cdot
    \end{pmatrix}.
\end{equation}
Found for $k=p=q=0$.

\subsection{$RC$}
\label{sec:RC}
\begin{equation}
\label{RC}
    RC=\begin{pmatrix}
        k^2&\cdot&\cdot&\cdot\\
        \cdot&k q&\cdot&\cdot\\
        \cdot&k^2-qp&k p&\cdot\\
        \cdot&\cdot&\cdot&k^2
    \end{pmatrix}. 
\end{equation}
\paragraph{Constant Model Intersections}
For $k^2=pq$ we find a special cases of $RA$ (\ref{sec:RA}), these settings are not discussed again.
\paragraph{\textcolor{RedOrange}{Non-constant Models}}
\begin{center}
\begin{tikzpicture}
   \draw[RedOrange!60, thick,  ->] (0,0) -- (4,1) ;
      \draw[RedOrange!60, thick,  ->] (0,0) -- (4,0) ;
       \draw[RedOrange!60, thick,  ->] (0,0) -- (4,-1) ;
       \draw[RedOrange!60, thick,  ->] (0,0) -- (4,-2) ;
       \draw[RedOrange!60, thick,  ->] (0,0) -- (4,2) ;
   \filldraw[color=RedOrange!70, fill=RedOrange!10,  thick ](0,0) circle (0.5);
     \draw (0,0) node {${RC}$};  
    \draw(4.5,1) circle (0.5);
    \draw (4.5,1) node {$\mathcal{R}^O$} ; 
         \draw(4.5,0) circle (0.5);
    \draw (4.5,0) node {$\mathcal{R}^{M'}$} ; 
      \draw(4.5,-1) circle (0.5);
    \draw (4.5,-1) node {$\mathcal{R}^{S}$} ;
     \draw(4.5,-2) circle (0.5);
    \draw (4.5,-2) node {$\mathcal{R}^{N'}$} ;
  \draw(4.5,2) circle (0.5);
    \draw (4.5,2) node {$\mathcal{R}^{E}$} ;
\end{tikzpicture} 
\end{center}

\paragraph{$\mathcal{R}^E$}
\begin{align}
\label{RE}
    \mathcal{R}^E:=\begin{pmatrix}
        k^2&\cdot&\cdot&\cdot\\
        \cdot&k q&\cdot&\cdot\\
        \cdot&e^{f(u)-f(v)}(k^2-pq)&k p&\cdot\\
        \cdot&\cdot&\cdot&k^2
    \end{pmatrix}.
\end{align}
Found for generic coefficients.
\paragraph{$\mathcal{R}^{N'}$}

\begin{align}
\label{RN'}
    \mathcal{R}^{N'}:=\begin{pmatrix}
        1&\cdot&\cdot&\cdot\\
        \cdot&1&\cdot&\cdot\\
        f(u)&-f(u)h(v)&\cdot&h(v)\\
        \cdot&\cdot&\cdot&1
    \end{pmatrix}.
    \end{align}
Found for $p=0$, $q=k=1$. (Similarly for $q=0$, $p=k=1$.)

\paragraph{$\mathcal{R}^{M'}$}
\begin{align}
\label{RQ}
    \mathcal{R}^{M'}:=\begin{pmatrix}
        e^{f(u)-f(v)}&\cdot&\cdot&\cdot\\
        \cdot&\cdot&\cdot&\cdot\\
        \cdot&e^{g(u)-g(v)}&\cdot&\cdot\\
        \cdot&\cdot&\cdot&e^{h(u)-h(v)}
    \end{pmatrix}.
    \end{align}
Found for $p=q=0$.

\paragraph{ $\mathcal{R}^{O}$, $\mathcal{R}^{S}$}
One finds \eqref{RO} and \eqref{RS} for $k=0$.

\subsection{$RC'$}
\label{sec:RC'}
\begin{equation}
     RC'=\begin{pmatrix}
        k^2&\cdot&\cdot&\cdot\\
        \cdot&q k &\cdot&\cdot\\
        \cdot&k^2-qp&pk&\cdot\\
        \cdot&\cdot&\cdot&-p q
    \end{pmatrix}.
\end{equation}

\paragraph{Constant Model Intersections}
For $k^2=pq$ we find a special cases of $RA$ (\ref{sec:RA}), these settings are not discussed again.
\paragraph{\textcolor{YellowOrange}{Non-constant Models}}
\begin{center}
\begin{tikzpicture}
 \draw[YellowOrange!90, very thick, ->] (0,0) -- (4,0) ;
    \filldraw[color=YellowOrange!100, fill=YellowOrange!25,  thick ](0,0) circle (0.5);
   \draw (0,0) node {${RC'}$};
   \draw(4.5,0) circle (0.5);
    \draw (4.5,0) node {$\mathcal{R}^{E'}$} ;  
\end{tikzpicture}
\end{center}

\paragraph{$\mathcal{R}^{E'}$}
\begin{align}
\label{RE'}
    \mathcal{R}^{E'}:=\begin{pmatrix}
        k^2&\cdot&\cdot&\cdot\\
        \cdot&k q(u)&\cdot&\cdot\\
        \cdot&e^{f(u)-f(v)}(k^2-p(u)q(u))&k p(v)&\cdot\\
        \cdot&\cdot&\cdot&-p(v)q(u)
    \end{pmatrix}.
\end{align}

\paragraph{$RD$}
\label{sec:RD}
\begin{equation}
    RD=\begin{pmatrix}
        \cdot&\cdot&\cdot&q\\
        \cdot&\cdot&k&\cdot\\
        \cdot&k&\cdot&\cdot\\
        p&\cdot&\cdot&\cdot
    \end{pmatrix}.
\end{equation}

\paragraph{\textcolor{Dandelion}{Non-constant models}}
\begin{center}
\begin{tikzpicture}
 \draw[Dandelion!90,very thick, ->] (0,0) -- (4,0) ;
  \draw[Dandelion!90, very thick, ->] (0,0) -- (4,-1) ;
    \draw[Dandelion!90, very thick, ->] (0,0) -- (4,1) ;
   \filldraw[color=Dandelion!100, fill=Dandelion!30,  thick ](0,0) circle (0.5);
   \draw(4.5,0) circle (0.5);
    \draw (0,0) node {${RD}$};
    \draw (4.5,0) node {$\mathcal{R}^B$} ;  
   \draw(4.5,-1) circle (0.5);
    \draw (4.5,-1) node {$\mathcal{R}^C$} ;
     \draw(4.5,1) circle (0.5);
    \draw (4.5,1) node {$\mathcal{R}^A$} ;

\end{tikzpicture} 
\end{center}

\paragraph{$\mathcal{R}^A$, $\mathcal{R}^C$ }
Models similar to \eqref{RA} or \eqref{RC} are found for $q=1/p$.
\paragraph{$\mathcal{R}^B$}
One finds \eqref{RB} for generic coefficients.

\subsection{$RE$}
\label{sec:RE}
\begin{equation}
   RE= \begin{pmatrix}
        k&-p&p&pq\\
        \cdot&k&\cdot&-q\\
        \cdot&\cdot&k&q\\
        \cdot&\cdot&\cdot&k
    \end{pmatrix}.
\end{equation}

\paragraph{Intersection of constant models}
For $p=q=0$ we find $RA$ (\ref{sec:RA}), this model is not discussed again.

\paragraph{\textcolor{Goldenrod}{Non-constant Models}}
\begin{center}
\begin{tikzpicture}
 \draw[Goldenrod!100,very thick, ->] (0,0) -- (4,0) ;
  \draw[Goldenrod!100, very thick, ->] (0,0) -- (4,-1) ;
  \draw[Goldenrod!100,very thick, ->] (0,0) -- (4,-2) ;
  \draw[Goldenrod!100,very thick, ->] (0,0) -- (4,2) ;
 \draw[Goldenrod!100,very thick, ->] (0,0) -- (4,1) ;
   \filldraw[color=Goldenrod!100, fill=Goldenrod!20,  thick ](0,0) circle (0.5);
   \draw(4.5,0) circle (0.5);
    \draw (0,0) node {${RE}$};
    \draw (4.5,0) node {$\mathcal{R}^K$} ;  
   \draw(4.5,-1) circle (0.5);
    \draw (4.5,-1) node {$\mathcal{R}^R$} ;
 \draw(4.5,1) circle (0.5);
    \draw (4.5,1) node {$\mathcal{R}^{H}$} ;
         \draw(4.5,-2) circle (0.5);
    \draw (4.5,-2) node {$\mathcal{R}^{N}$} ;
     \draw(4.5,2) circle (0.5);
    \draw (4.5,2) node {$\mathcal{R}^{G}$} ;
\end{tikzpicture}
\end{center}

\paragraph{$\mathcal{R}^H$}
\begin{equation}
\label{RH}
   \mathcal{R}^H:=
    \begin{pmatrix}
        1&-p&p& p q + f(u)-f(v) \\
        \cdot &1 &\cdot & -q \\
          \cdot &\cdot &1 & q \\
            \cdot &\cdot &\cdot & 1
    \end{pmatrix}.
\end{equation}
Found for $k=1$.

\paragraph{$\mathcal{R}^G$}
One finds \eqref{RG} for $q=-p$ or for $k=0$.

\paragraph{$  \mathcal{R}^{N}$}
\label{RN}
\begin{equation}
   \mathcal{R}^{N}:= \begin{pmatrix}
        1&\cdot&\cdot&f(v)h(u)\\
        \cdot&1&\cdot&-h(u)\\
        \cdot&\cdot&1&f(v)\\
        \cdot&\cdot&\cdot&\cdot
    \end{pmatrix}.
\end{equation}
Found for $k=p=0$.

\paragraph{$\mathcal{R}^{R}$}
One finds \eqref{RR} for $q=k=0$.

\paragraph{$\mathcal{R}^{K}$}
One finds \eqref{RK} for $k=0$ .

\subsection{$RF$}
\begin{equation}
   RF= \begin{pmatrix}
        p&\cdot&\cdot&k\\
        \cdot&q&\cdot&\cdot\\
        \cdot&p-q&p&\cdot\\
        \cdot&\cdot&\cdot&-q
    \end{pmatrix}.
\end{equation}

\paragraph{Intersections of constant models}
For ($p=q$ and $k=0$) we find a special case of $RA$ (\ref{sec:RA}). For ($k=0$) one finds special cases of $RC'$ (\ref{sec:RC'}). For $p=q=0$ we find the special case $RB(k=p=q=0)$ (\ref{sec:RB}). These cases are not discussed again.

\paragraph{\textcolor{GreenYellow}{Non-constant models}}
\begin{center}
    \begin{tikzpicture}
 \draw[GreenYellow!100, very thick, ->] (0,0) -- (4,0) ;
  \draw[GreenYellow!100, very thick, ->] (0,0) -- (4,1) ;
   \filldraw[color=GreenYellow!100, fill=GreenYellow!30,  thick ](0,0) circle (0.5);
    \draw (0,0) node {${RF}$};
     \draw(4.5,0) circle (0.5);
     \draw (4.5,0) node {$\mathcal{R}^F$} ;
     \draw(4.5,1) circle (0.5);
     \draw (4.5,1) node {$\mathcal{R}^{I}$} ;   
\end{tikzpicture}
\end{center}

\paragraph{$\mathcal{R}^I$}
\begin{equation}
\label{RI}
   \mathcal{R}^I:= \begin{pmatrix}
        1&\cdot&\cdot&k(u) e^{u-v}\\
        \cdot&-1&\cdot&\cdot\\
        \cdot&2 e^{u-v}&1&\cdot\\
        \cdot&\cdot&\cdot&1
    \end{pmatrix}.
\end{equation}
Found for $q=-p=-1$.

\paragraph{$\mathcal{R}^F$}
One finds \eqref{RF} for $p=q=1$.

\subsection{$RG$}
\label{sec:RG}
\begin{equation}
   RG=\begin{pmatrix}
        p^2+2pq-q^2&\cdot&\cdot&p^2-q^2\\
        \cdot&p^2+q^2&p^2-q^2&\cdot\\
        \cdot&p^2-q^2&p^2+q^2&\cdot\\
        p^2-q^2&\cdot &\cdot&p^2-2q p-q^2
    \end{pmatrix}.
\end{equation}

\paragraph{Intersections of constant models}
Putting $p=q$ leads to $RA$ (\ref{sec:RA}). This setting will not be discussed again.

\paragraph{\textcolor{SpringGreen}{Non-constant models}}
\begin{center}
    \begin{tikzpicture}
 \draw[SpringGreen!100, very thick, ->] (0,0) -- (4,0) ;
   \draw[SpringGreen!100,very  thick, ->] (0,0) -- (4,1) ;
   \filldraw[color=SpringGreen!100, fill=SpringGreen!40,  thick ](0,0) circle (0.5);
   \draw(4.5,0) circle (0.5);
    \draw (0,0) node {${RG}$};
    \draw (4.5,0) node {$\mathcal{R}^B$} ;  
     \draw(4.5,1) circle (0.5);
    \draw (4.5,1) node {$\mathcal{R}^{A}$} ;
\end{tikzpicture}
\end{center}

\paragraph{$ \mathcal{R}^{B}$}
For $p=0$ one finds a model similar to \eqref{RB}.
\paragraph{$\mathcal{R}^A$}
For $q=0$ we find models similar to \eqref{RA}.

\subsection{$RH$}
\begin{equation}
   RH= \begin{pmatrix}
        1&\cdot&\cdot&1\\
        \cdot&-1&1&\cdot\\
        \cdot&1&1&\cdot\\
        -1&\cdot&\cdot&1
    \end{pmatrix}.
\end{equation}

\paragraph{\textcolor{LimeGreen}{Non-constant Model}}
\begin{center}
\begin{tikzpicture}
 \draw[LimeGreen!100, very thick, ->] (0,0) -- (4,0) ;
   \filldraw[color=LimeGreen!100, fill=LimeGreen!30,  thick ](0,0) circle (0.5);
   \draw(4.5,0) circle (0.5);
    \draw (0,0) node {${RH}$};
    \draw (4.5,0) node {$\mathcal{R}^J$} ;  
\end{tikzpicture}
\end{center}

\paragraph{$\mathcal{R}^J$}
\begin{equation}
\label{RJ}
  \mathcal{R}^J:= \begin{pmatrix}
        1&\cdot&\cdot&e^{u-v}\\
        \cdot&-1&e^{u-v}&\cdot\\
        \cdot&e^{u-v}&1&\cdot\\
        -e^{u-v}&\cdot&\cdot&1
    \end{pmatrix}.
\end{equation}

\subsection{$RI$}
\begin{equation}
   RI= \begin{pmatrix}
        1&\cdot&\cdot&1\\
        \cdot&-1&\cdot&\cdot\\
        \cdot&\cdot&-1&\cdot\\
        \cdot&\cdot&\cdot&1
    \end{pmatrix}.
\end{equation}

\paragraph{\textcolor{OliveGreen}{Non-constant Models}}
\begin{center}
\begin{tikzpicture}
 \draw[OliveGreen!100, thick, ->] (0,0) -- (4,0) ;
   \filldraw[color=OliveGreen!100, fill=OliveGreen!20,  thick ](0,0) circle (0.5);
   \draw(4.5,0) circle (0.5);
    \draw (0,0) node {${RI}$};
    \draw (4.5,0) node {$\mathcal{R}^D$} ;  
\end{tikzpicture}
\end{center}

\subsection{$RJ$}
\label{sec:RJ}
\begin{equation}
   RJ=\begin{pmatrix}
        p+q&\cdot&\cdot&\cdot\\
        \cdot&q&\cdot&q\\
        \cdot&\cdot &p+q&\cdot\\
        \cdot&p &\cdot&p
    \end{pmatrix}.
    \end{equation}
    
\paragraph{\textcolor{ForestGreen}{Non-constant Models}}
\begin{center}
    \begin{tikzpicture}
 \draw[ForestGreen!100, thick, ->] (0,0) -- (4,0) ;
   \filldraw[color=ForestGreen!100, fill=ForestGreen!20,  thick ](0,0) circle (0.5);
   \draw(4.5,0) circle (0.5);
    \draw (0,0) node {${RJ}$};
    \draw (4.5,0) node {$\mathcal{R}^{O}$} ;  
\end{tikzpicture}
\end{center}

\paragraph{$\mathcal{R}^{O}$}
One finds \eqref{RO} for $p=-q$.

\subsection{$RK$}
\label{sec:RK}
    \begin{equation}
   RK= \begin{pmatrix}
        \cdot&p&p&\cdot \\
        \cdot&\cdot&k&q\\
        \cdot&k&\cdot&q\\
        \cdot&\cdot&\cdot&\cdot
    \end{pmatrix}.
    \end{equation}
    
    \paragraph{Intersection of constant models}
    For $p=q=0$ one finds a special case of $RD$ (\ref{sec:RD}).
    
\paragraph{\textcolor{SeaGreen}{Non-constant Models}}
\begin{center}
    \begin{tikzpicture}
 \draw[SeaGreen!90, very thick, ->] (0,0) -- (4,0) ;
  \draw[SeaGreen!90,very thick, ->] (0,0) -- (4,-1) ;
   \draw[SeaGreen!90,very thick, ->] (0,0) -- (4,1) ;
   \filldraw[color=SeaGreen!90, fill=SeaGreen!25,  thick ](0,0) circle (0.5);
    \draw (0,0) node {${RK}$};
    \draw(4.5,1) circle (0.5);
    \draw (4.5,1) node {$\mathcal{R}^G$} ;
    \draw(4.5,0) circle (0.5);
    \draw (4.5,0) node {$\mathcal{R}^K$} ;
     \draw(4.5,-1) circle (0.5);
    \draw (4.5,-1) node {$\mathcal{R}^R$} ;
\end{tikzpicture}
\end{center}

\paragraph{$\mathcal{R}^K$}
 \begin{equation}
 \label{RK}
  \mathcal{R}^K:= \begin{pmatrix}
        \cdot&p &p e^{g(u)-g(v)}&f(u,v)\\
        \cdot&\cdot&k e^{g(u)-g(v)}&q e^{g(u)-g(v)}\\
        \cdot&k &\cdot&q \\
        \cdot&\cdot&\cdot&\cdot
    \end{pmatrix}.
    \end{equation}
    Found for generic coefficients.
    
    \paragraph{$\mathcal{R}^R$}
   One finds \eqref{RR} for $p=k=0$ (or when transposed for $q=k=0$).
    
    \paragraph{$\mathcal{R}^A$}
    One finds a model similar to \eqref{RA} can be found for $k=0$.
    
\subsection{$RL$}
\label{sec:RL}
 \begin{equation}
   RL= \begin{pmatrix}
        1&\cdot&\cdot&\cdot\\
        \cdot&\cdot&\cdot&1\\
        \cdot&1&\cdot&\cdot\\
        \cdot&\cdot&\cdot&1
    \end{pmatrix}.
    \end{equation}
    
\paragraph{\textcolor{Aquamarine}{Non-constant solutions}}
\begin{center}
    \begin{tikzpicture}
 \draw[Aquamarine!100, very thick, ->] (0,0) -- (4,0) ;
   \filldraw[color=Aquamarine!100, fill=Aquamarine!20,  thick ](0,0) circle (0.5);
   \draw(4.5,0) circle (0.5);
    \draw (0,0) node {${RL}$};
    \draw (4.5,0) node {$\mathcal{R}^L$} ;  
\end{tikzpicture}
\end{center}

\paragraph{$\mathcal{R^L}$}
  \begin{equation}
  \label{RL}
   \mathcal{R^L}(u,v):= \begin{pmatrix}
        1&\cdot&\cdot&\cdot\\
        \cdot&\cdot&\cdot&1\\
        \cdot&e^{u-v}&\cdot&1-e^{u-v}\\
        \cdot&\cdot&\cdot&1
    \end{pmatrix}.
    \end{equation}
    
\subsection{$RM$}
\label{sec:RM}
 \begin{equation}
     RM= \begin{pmatrix}
        1&\cdot&\cdot&\cdot\\
        \cdot&\cdot &1&\cdot\\
        \cdot&1&\cdot &\cdot\\
        \cdot&\cdot&\cdot&\cdot
    \end{pmatrix}.
\end{equation}

\paragraph{\textcolor{TealBlue}{Non-constant Models}}
\begin{center}
    \begin{tikzpicture}
 \draw[TealBlue!100, thick, ->] (0,0) -- (4,0) ;
   \filldraw[color=TealBlue!100, fill=TealBlue!20,  thick ](0,0) circle (0.5);
   \draw(4.5,0) circle (0.5);
    \draw (0,0) node {${RM}$};
    \draw (4.5,0) node {$\mathcal{R}^M$} ;  
\end{tikzpicture}
\end{center}

\paragraph{$\mathcal{R}^M$}
    \begin{equation}
    \label{RM}
   \mathcal{R}^M(u,v):= \begin{pmatrix}
        1&\cdot&\cdot&\cdot\\
        \cdot&\cdot&e^{g(u)-g(v)}&\cdot\\
        \cdot&e^{f(u)-f(v)}&\cdot&\cdot\\
        \cdot&\cdot&\cdot&\cdot
    \end{pmatrix}.
    \end{equation}

\subsection{$RN$}
\label{sec:RN}
    \begin{equation}
     RN= \begin{pmatrix}
        \cdot&(q-k)(p^2-k^2)&(q+k)(p^2-k^2)&s\\
        \cdot&\cdot &\cdot&(q^2-k^2)(p+k)\\
        \cdot&\cdot&\cdot &(q^2-k^2)(p-k)\\
        \cdot&\cdot&\cdot&\cdot
    \end{pmatrix}.
\end{equation}

\paragraph{Intersection of constant models}
$(p=q=0)=(p=q,k=0)$ lead back to  $RB(k=0,p=-q)$ (\ref{sec:RB}). $(k=0,p=q)$ leads back to  $RB(k=0,p=q)$ (\ref{sec:RB}).
($s=k=0$) leads back to se $RK(k=0)$ (\ref{sec:RK}).
For $k=p=q$ we find $RB(k=p=q=0)$ (\ref{sec:RB}).
Setting either $p$ or $q$ equal to either $k$ or $-k$, leads to $RP(k=0)$ (\ref{sec:RP})(up to rotation and or permutation).
\\
These cases will not be discussed again.

\paragraph{\textcolor{Cyan}{Non-constant Models}}
\begin{center}
    \begin{tikzpicture}
 \draw[Cyan!90, thick, ->] (0,0) -- (4,0) ;
 \draw[Cyan!90, thick, ->] (0,0) -- (4,-1) ;
  \draw[Cyan!90, thick, ->] (0,0) -- (4,-2) ;
 \draw[Cyan!90, thick, ->] (0,0) -- (4,1) ;
   \filldraw[color=Cyan!100, fill=Cyan!20, very thick ](0,0) circle (0.5);
    \draw (0,0) node {${RN}$};
       \draw(4.5,0) circle (0.5);
    \draw (4.5,0) node {$\mathcal{R}^{G}$} ;  
      \draw(4.5,-1) circle (0.5);
    \draw (4.5,-1) node {$\mathcal{R}^{K}$} ;  
 \draw(4.5,1) circle (0.5);
    \draw (4.5,1) node {$\mathcal{R}^{D}$} ;
          \draw(4.5,-2) circle (0.5);
    \draw (4.5,-2) node {$\mathcal{R}^{Q}$} ;  
\end{tikzpicture}
\end{center}

\paragraph{$\mathcal{R}^Q$}
 \begin{equation}
     \mathcal{R}^Q= \begin{pmatrix}
        \cdot&(q-k)(p^2-k^2)&(q+k)(p^2-k^2)&f(u,v)\\
        \cdot&\cdot &\cdot&(q^2-k^2)(p+k)\\
        \cdot&\cdot&\cdot &(q^2-k^2)(p-k)\\
        \cdot&\cdot&\cdot&\cdot
    \end{pmatrix}.
\end{equation}
Found for generic coefficients.

\paragraph{$\mathcal{R}^{D}$}
One finds a models similar to \eqref{RD} for  $(k=0, p=-q)$.

\paragraph{$\mathcal{R}^{K}$}
One finds \eqref{RK} for  $(k=0)$.

\paragraph{$\mathcal{R}^{G}$}
One finds \eqref{RK} for  $(k=0)$.

\subsection{$RO$}
\label{sec:RO}
 \begin{equation}
     RO= \begin{pmatrix}
        1&1&1&\cdot\\
        \cdot&\cdot &\cdot&\cdot\\
        \cdot&\cdot&\cdot &\cdot\\
        \cdot&\cdot&\cdot&1
    \end{pmatrix}.
\end{equation}
No non-constant models.

\subsection{$RP$}
\label{sec:RP}
    \begin{equation}
     RP= \begin{pmatrix}
        \cdot&p&\cdot&q\\
        \cdot&\cdot &\cdot&\cdot\\
        \cdot&k&\cdot &\cdot\\
        \cdot&\cdot&\cdot&\cdot
    \end{pmatrix}.
\end{equation}

\paragraph{Intersection of constant models}
For $p=k=0$ we find the special case $RB(k=q=p=0)$ (\ref{sec:RB}).  For $p=q=0$ we find the special case $RC(k=0)$ (\ref{sec:RC}).
These cases are not discussed again.

\paragraph{\textcolor{RoyalBlue}{Non-constant Models}}
\begin{center}
    \begin{tikzpicture}
 \draw[RoyalBlue!90, very thick, ->] (0,0) -- (4,0) ;
  \draw[RoyalBlue!90, very thick, ->] (0,0) -- (4,-1) ;
   \draw[RoyalBlue!90, very  thick, ->] (0,0) -- (4,1) ;
   \filldraw[color=RoyalBlue!100, fill=RoyalBlue!10, very thick ](0,0) circle (0.5);
   \draw(4.5,0) circle (0.5);
    \draw (0,0) node {${RP}$};
    \draw (4.5,0) node {$\mathcal{R}^P$} ;  
   \draw(4.5,-1) circle (0.5);
    \draw (4.5,-1) node {$\mathcal{R}^R$} ;
    \draw(4.5,1) circle (0.5);
    \draw (4.5,1) node {$\mathcal{R}^O$} ;
\end{tikzpicture}
\end{center}

\paragraph{$\mathcal{R}^P$}
 \begin{equation}
 \label{RP}
    \mathcal{R}^P:= \begin{pmatrix}
        \cdot&f_1(u,v)&\cdot&f_2(u,v)\\
        \cdot&\cdot &\cdot&f_3(u,v)\\
        \cdot&\cdot&\cdot &\cdot\\
        \cdot&\cdot&\cdot&\cdot
    \end{pmatrix}.
\end{equation}
Found for $k=0$.

\paragraph{$\mathcal{R}^R$}
One finds \eqref{RR} for $k=0$.

\paragraph{$\mathcal{R}^O$}
One finds \eqref{RO} for generic coefficients.

\subsection{$RQ$}
\label{sec:RQ}
    \begin{equation}
     RQ= \begin{pmatrix}
        \cdot&p&\cdot&\cdot\\
        \cdot&\cdot &\cdot&q\\
        \cdot&\cdot&\cdot &\cdot\\
        \cdot&\cdot&\cdot&\cdot
    \end{pmatrix}.
\end{equation}

\paragraph{Intersection of constant Models}
For $q=0$ one finds the special case $RP(q=k=0)$ (\ref{sec:RP}). Similarly, setting  $p=0$ is $RP(q=k=0)$ (\ref{sec:RP}) up to a permutation and a rotation around the off diagonal.\\
These cases are not discussed again.

\paragraph{\textcolor{Periwinkle}{Non-constant Models}}
\begin{center}
    \begin{tikzpicture}
 \draw[Periwinkle!90, very thick, ->] (0,0) -- (4,0) ;   \filldraw[color=Periwinkle!100, fill=Periwinkle!20, thick ](0,0) circle (0.5);
    \draw (0,0) node {${RQ}$}; 
   \draw(4.5,0) circle (0.5);
    \draw (4.5,0) node {$\mathcal{R}^P$} ;
\end{tikzpicture}
\end{center}

\paragraph{$\mathcal{R}^P$}
One finds \eqref{RP} generically.

\subsection{$RR$}
 \label{sec:RR}
    \begin{equation}
     RR= \begin{pmatrix}
        \cdot&p&\cdot&\cdot\\
        \cdot&\cdot &\cdot&\cdot\\
        \cdot&\cdot&\cdot &q\\
        \cdot&\cdot&\cdot&\cdot
    \end{pmatrix}.
\end{equation}

\paragraph{Intersection of constant models}
For $(q=0)$ we find the special case $RP(p=k=0)$ (\ref{sec:RP}). By rotation around the off diagonal, which is an allowed similarity transformation, we find $RR(p=0)$ (\ref{sec:RR}). These cases will not be discussed again.

\paragraph{\textcolor{Orchid}{Non-constant Models}}
\begin{center}
    \begin{tikzpicture}
 \draw[Orchid!90, very thick, ->] (0,0) -- (4,0) ;
 \draw[Orchid!90, very thick, ->] (0,0) -- (4,1) ;
   \filldraw[color=Orchid!100, fill=Orchid!20, very thick ](0,0) circle (0.5);
   \draw(4.5,0) circle (0.5);
    \draw (0,0) node {${RR}$};
    \draw (4.5,0) node {$\mathcal{R}^O$} ;  
   \draw(4.5,1) circle (0.5);
    \draw (4.5,1) node {$\mathcal{R}^G$} ;
\end{tikzpicture}
\end{center}

\paragraph{$\mathcal{R}^G$}
One finds \eqref{RG} for generic coefficients.

\paragraph{$\mathcal{R}^O$}
One finds \eqref{RO} for generic coefficients.

\subsection{$RS$}
\label{sec:RS}
    \begin{equation}
     R^S= \begin{pmatrix}
        \cdot&p&q&\cdot\\
        \cdot&\cdot &\cdot&\cdot\\
        \cdot&\cdot&\cdot &\cdot\\
        \cdot&\cdot&\cdot&\cdot
    \end{pmatrix}.
\end{equation}

\paragraph{Intersections of constant Models}
Note that permutation of this matrix, simply switches $p$ and $q$. For $q=0$ ($p=0$) we find the special case $RP(p=k=0)$ (same but permuted) (\ref{sec:RP}).
For $p=-q$ we find the special case $RE(k=q=0)$ (\ref{sec:RE}).\\
These cases will not be discussed again.

\paragraph{\textcolor{Thistle}{Non-constant Models}}
\begin{center}
\begin{tikzpicture}
 \draw[Thistle!90, very thick, ->] (0,0) -- (4,0) ;
   \filldraw[color=Thistle!100, fill=Thistle!30, very thick ](0,0) circle (0.5);
    \draw (0,0) node {${RS}$};
    \draw(4.5,0) circle (0.5);
    \draw (4.5,0) node {$\mathcal{R}^R$} ;
\end{tikzpicture}
\end{center}

\paragraph{$\mathcal{R}^R$}
One finds \eqref{RR} for generic coefficients.

\subsection{$RT$}
\label{sec:RT}
   \begin{equation}
     R^T= \begin{pmatrix}
        \cdot&\cdot&\cdot&\cdot\\
        \cdot&p &q&\cdot\\
        \cdot&\cdot&\cdot &\cdot\\
        \cdot&\cdot&\cdot&\cdot
    \end{pmatrix}.
\end{equation}

\paragraph{Intersection of constant models}
Setting $q=0$ one find a special case of $RA$ (\ref{sec:RA}). Setting $p=0$ one find a special case of $RC(k=0)$ (\ref{sec:RC})

\paragraph{\textcolor{Lavender}{Non-constant Models}}
\begin{center}
    \begin{tikzpicture}
 \draw[Lavender!100, very  thick, ->] (0,0) -- (4,0) ;
   \filldraw[color=Lavender!100, fill=Lavender!20, very thick ](0,0) circle (0.5);
   \draw(4.5,0) circle (0.5);
    \draw (0,0) node {${RT}$};
    \draw (4.5,0) node {$\mathcal{R}^S$} ;   
\end{tikzpicture}
\end{center}

\paragraph{$\mathcal{R}^S$}
One finds \eqref{RS} for generic coefficients.

\section{$\tilde{\mathcal{R}}$}
The most general form that solves the fundamental commutation relation for the non-regular solutions of the Yang--Baxter equation. We present them for generic values of the parameters.

\subsection{$\tilde{\mathcal{R}}_{c}$}
\begin{align}
 \tilde{\mathcal{R}}_c(u,v)=
\begin{pmatrix}
        g_1(u,v)&\cdot&\cdot&g_3(u,v)\\
        \cdot&g_2(u,v)&\frac{f(u,0)}{f(v,0)}(g_1(u,v)+g_2(u,v))&\cdot\\
        \cdot&\frac{f(v,0)}{f(u,0)} (g_1(u,v) + g_2(u,v))&g_2(u,v)&\cdot\\
       \cdot&\cdot&\cdot&g_1(u,v)
    \end{pmatrix}.
\end{align}

\subsection{$\tilde{\mathcal{R}}_{d}$}
\begin{align}
    \tilde{\mathcal{R}}_d(u,v)=
   \begin{pmatrix}
        g_1(u,v)&\cdot&\cdot&\cdot\\
        \cdot&g_2(u,v)&e^{-u+v}(g_1(u,v)-\frac{g_2(u,v)}{q})&\cdot\\
        \cdot&g_3(u,v)&\frac{1}{q}(g_1(u,v)-e^{-u+v}g_3(u,v))&\cdot\\
       \cdot&\cdot&\cdot&-\frac{g_2(u,v)}{q}+e^{-u+v}g_3(u,v)
    \end{pmatrix}.
\end{align}

\subsection{$\tilde{\mathcal{R}}_{e}$}
\begin{align}
    \tilde{\mathcal{R}}_e(u,v)=
   \begin{pmatrix}
        g_1(u,v)&\cdot&\cdot&g_4(u,v)\\
        \cdot&g_2(u,v)&g_3(u,v)&\cdot\\
        \cdot&\frac{f(v,0)}{f(u,0)} (g_1(u,v) - g_2(u,v))&g_1(u,v)-\frac{f(v,0)}{f(u,0)}g_3(u,v)&\cdot\\
       \cdot&\cdot&\cdot&-g_2(u,v)+\frac{f(v,0)}{f(u,0)}g_3(u,v)
    \end{pmatrix}.
\end{align}

\subsection{$\tilde{\mathcal{R}}_{f}$}
\begin{equation}
   \tilde{\mathcal{R}}_f(u,v)= \begin{pmatrix}
        f_1(u,v)&\cdot&\cdot&f_4(u,v)\\
        \cdot&f_2(u,v)&f_3(u,v)&\cdot\\
        \cdot&F_1&F_2&\cdot\\
        F_3&\cdot&\cdot&F_4
    \end{pmatrix},
\end{equation}
\begin{equation}
    \begin{aligned}
        F_1=\frac{2}{k}f_4(u,v)+ e^{v-u}(f_1(u,v)+f_2(u,v)),\\
        F_2=-f_1(u,v) +e^{v-u}(f_3(u,v)+\frac{2}{k}f_4(u,v)),\\
        F_3=\frac{2}{k}(e^{v-u}(f_1(u,v)+f_2(u,v))-f_3(u,v)),\\
        F_4=f_2(u,v)+e^{v-u}(\frac{2}{k}f_4(u,v)+f_3(u,v)).
    \end{aligned}
\end{equation}

\subsection{$\tilde{R}_{g}$}
\begin{equation}
   \tilde{R}_g(u,v)= \left(
\begin{array}{cccc}
 f_1(u,v) & 0 & 0 & f_2(u,v) \\
 0 & H_1(u,v) & H_2(u,v) & 0 \\
 0 & H_2 (u,v)&-H_1(u,v) & 0 \\
 -f_2(u,v) & 0 & 0 & f_1(u,v) \\
\end{array}
\right),
\end{equation}
\begin{equation}
\begin{aligned}
   H_1(u,v)=\frac{e^{2 u} f_1(u,v)-e^{2 v} f_1(u,v)-2 e^{u+v} f_2(u,v)}{e^{2 u}+e^{2 v}}, \\
    H_2(u,v)=\frac{2 e^{u+v} f_1(u,v)+e^{2 u} f_2(u,v)-e^{2 v} f_2(u,v)}{e^{2 u}+e^{2 v}}.
\end{aligned}
\end{equation}

\subsection{$\tilde{\mathcal{R}}_{k}$}
\begin{equation}
    \tilde{\mathcal{R}}_k(u,v)=\begin{pmatrix}
       g_1(u,v) & g_2(u,v) & g_3(u,v) & G_5(u,v), \\
 g_4(u,v) & G_6(u,v) & g_6(u,v) & g_5(u,v) ,\\
 g_7(u,v) & g_8(u,v) & G_7(u,v) & g_{9}(u,v), \\
G_1(u,v) &G_2(u,v) & G_3(u,v) & G_4(u,v)
    \end{pmatrix}.
\end{equation}
\begin{equation}
\begin{aligned}
    G_1(u,v)=  -g_1(u,v)-g_2(u,v)-g_7(u,v), \\
    G_2(u,v)= -g_2(u,v)-g_4(u,v)+g_5(u,v)-g_6(u,v)-g_8(u,v),\\
    G_3(u,v)=-g_3(u,v)-g_5(u,v)-g_7(u,v)+g_8(u,v)-g_9(u,v),\\
   G_4(u,v)= g_1(u,v)-g_2(u,v)-g_3(u,v)-g_4(u,v)-g_9(u,v),\\
 G_5(u,v)  -g_1(u,v)+g_2(u,v)+g_3(u,v),\\
 G_6(u,v)=g_4(u,v)-g_5(u,v)+g_4(u,v),\\
 G_7(u,v)=g_7(u,v)-g_8(u,v)+g_{9}(u,v).
\end{aligned}
\end{equation}

\subsection{$\tilde{\mathcal{R}}_{i}$}
\begin{equation}
   \tilde{\mathcal{R}}_{i}(u,v)=
\begin{pmatrix} 
1& 0 & 0 & 0 \\
 0 & 0 &  g_1(u,v) & 1-g_1(u,v) \\
 0 & e^{u-v} & 0 & 1-e^{u-v}\\
 0 & 0 & 0 & 1\\
\end{pmatrix}.
\end{equation}

\subsection{$\tilde{\mathcal{R}}_{j}$}
\begin{equation}
   \tilde{\mathcal{R}}_{j}(u,v)=
\begin{pmatrix} 
1& 0 & 0 & 0 \\
 0 & 0 &  e^{u-v} & 0 \\
 0 & e^{u-v}& 0 & 0 \\
 0 & 0 & 0 & g_1(u,v) \\
\end{pmatrix}.
\end{equation}

\subsection{$\tilde{\mathcal{R}}_{l}$}
Any general matrix solves the Lax equation for $\mathcal{R}_l$.

\subsection{$\tilde{\mathcal{R}}_{m}$}
\begin{equation}
\tilde{\mathcal{R}}_{n}=
    \begin{pmatrix} 
1& g_1(u,v) & g_2(u,v) & g_3(u,v) \\
 0 & g_4(u,v)& g_5(u,v) & g_6(u,v) \\
 0 &\frac{f_1(u,0) f_3(v,0) }{f_1(v,0) f_3(u,0)} & 0 & -\frac{ ({f_2}(v,0) {f_3}(u,0)-{f_2}(u,0) {f_3}(v,0))}{{f_1}(v,0) {f_3}(u,0)} \\
 0 & 0 & 0 & 1\\
 \end{pmatrix}.
\end{equation}

\subsection{$\tilde{\mathcal{R}}_{m}$}
\begin{equation}
\tilde{\mathcal{R}}_{m}=
    \begin{pmatrix} 
g_1(u,v)& g_2(u,v)& g_3(u,v)& g_4(u,v)\\
 0 & g_5(u,v) & g_1(u,v)-\frac{(k-p) g_5(u,v)}{k-q} & g_6(u,v) \\
 0 & g_1(u,v)-\frac{(k-p) g_5(u,v)}{k-q} & \frac{(k-p)^2 g_5(u,v)}{(k-q)^2} & F(u,v) \\
 0 & 0 & 0 & g_1(u,v) \\
 \end{pmatrix}.
\end{equation}
\begin{equation}
F(u,v)=-\frac{g_1(u,v) (p f_3(u)-p f_3(v)-q f_3(u)+q f_3(v))}{(k-p) (k-q)}+\frac{(k-p) (k+q) g_2(u,v)}{(k+p) (k-q)}-\frac{(-k-q) g_3(u,v)}{k+p}-\frac{(k-p) g_6(u,v)}{k-q}.
\end{equation}

\subsection{$\tilde{\mathcal{R}}_{o}$}
\begin{equation}
\tilde{\mathcal{R}}_{o}=
    \begin{pmatrix} 
1& g_1(u,v) & g_2(u,v) & g_3(u,v) \\
 0 & 0&\frac{f_1(u,0) f_3(v,0) }{f_1(v,0) f_3(u,0)} & -\frac{ ({f_2}(v,0) {f_3}(u,0)-{f_2}(u,0) {f_3}(v,0))}{{f_1}(v,0) {f_3}(u,0)} \\
 0 & g_4(u,v)& g_5(u,v) & g_6(u,v) \\
 0 & 0 & 0 & 1\\
 \end{pmatrix}.
\end{equation}

\begin{bibtex}[\jobname]

@misc{perk2006yangbaxterequations,
      title={Yang-Baxter Equations}, 
      author={Jacques H. H. Perk and Helen Au-Yang},
      year={2006},
      eprint={math-ph/0606053},
      archivePrefix={arXiv},
      primaryClass={math-ph},
      url={https://arxiv.org/abs/math-ph/0606053}, 
}

@article{Phys,
  title = {Some Exact Results for the Many-Body Problem in one Dimension with Repulsive Delta-Function Interaction},
  author = {Yang, C. N.},
  journal = {Phys. Rev. Lett.},
  volume = {19},
  issue = {23},
  pages = {1312--1315},
  numpages = {0},
  year = {1967},
  month = {Dec},
  publisher = {American Physical Society},
  doi = {10.1103/PhysRevLett.19.1312},
  url = {https://link.aps.org/doi/10.1103/PhysRevLett.19.1312}
}
@article{BAXTER,
title = {Partition function of the Eight-Vertex lattice model},
journal = {Annals of Physics},
volume = {70},
number = {1},
pages = {193-228},
year = {1972},
issn = {0003-4916},
doi = {https://doi.org/10.1016/0003-4916(72)90335-1},
url = {https://www.sciencedirect.com/science/article/pii/0003491672903351},
author = {Rodney J Baxter},
abstract = {The partition function of the zero-field “Eight-Vertex” model on a square M by N lattice is calculated exactly in the limit of M, N large. This model includes the dimer, ice and zero-field Ising, F and KDP models as special cases. In general the free energy has a branch point singularity at a phase transition, with an irrational exponent.}
}

@article{Maity:2024vtr,
    author = "Maity, Somnath and Singh, Vivek Kumar and Padmanabhan, Pramod and Korepin, Vladimir",
    title = "{Hietarinta's classification of $4\times 4$ constant Yang-Baxter operators using algebraic approach}",
    eprint = "2409.05375",
    archivePrefix = "arXiv",
    primaryClass = "hep-th",
    month = "9",
    year = "2024"
}

@article{Corcoran:2023zax,
    author = "Corcoran, Luke and de Leeuw, Marius",
    title = "{All regular $4 \times 4$ solutions of the Yang-Baxter equation}",
    eprint = "2306.10423",
    archivePrefix = "arXiv",
    primaryClass = "hep-th",
    month = "6",
    year = "2023"
}

@article{HIETARINTA,
title = {All solutions to the constant quantum Yang-Baxter equation in two dimensions},
journal = {Physics Letters A},
volume = {165},
number = {3},
pages = {245-251},
year = {1992},
issn = {0375-9601},
doi = {https://doi.org/10.1016/0375-9601(92)90044-M},
url = {https://www.sciencedirect.com/science/article/pii/037596019290044M},
author = {Jarmo Hietarinta},
abstract = {In this Letter we present an exhaustive list of solutions to the constant quantum Yang-Baxter equation Rk1k2j1j2Rl1k3k1j3Rl2l3k2k3=Rk2k3j2j3Rk1l3j1k3Rl1l2k1k2 in two dimensions (i.e. all indices ranging over 1, 2 with summation over repeated indices). This set of 64 equations for 16 unknowns was first reduced by hand to manageable subcases which were then solved by computer using Gröbner basis methods. For the presentation of the results we propose a canonical form based on the various trace matrices of R.}
}

@article{Garkun:2024jnp,
    author = "Garkun, Alexander. S. and Barik, Suvendu K. and Fedorov, Aleksey K. and Gritsev, Vladimir",
    title = "{New spectral-parameter dependent solutions of the Yang-Baxter equation}",
    eprint = "2401.12710",
    archivePrefix = "arXiv",
    primaryClass = "quant-ph",
    month = "1",
    year = "2024"
}

@article{deLeeuw:2020ahe,
    author = "de Leeuw, Marius and Paletta, Chiara and Pribytok, Anton and Retore, Ana L. and Ryan, Paul",
    title = "{Classifying Nearest-Neighbor Interactions and Deformations of AdS}",
    eprint = "2003.04332",
    archivePrefix = "arXiv",
    primaryClass = "hep-th",
    doi = "10.1103/PhysRevLett.125.031604",
    journal = "Phys. Rev. Lett.",
    volume = "125",
    number = "3",
    pages = "031604",
    year = "2020"
}

@article{Vieira_2018,
   title={Solving and classifying the solutions of the Yang-Baxter equation through a differential approach. Two-state systems},
   volume={2018},
   ISSN={1029-8479},
   url={http://dx.doi.org/10.1007/JHEP10(2018)110},
   DOI={10.1007/jhep10(2018)110},
   number={10},
   journal={Journal of High Energy Physics},
   publisher={Springer Science and Business Media LLC},
   author={Vieira, R. S.},
   year={2018},
   month=oct }

@book{jimbo1990yang,
  title={Yang-Baxter Equation in Integrable Systems},
  author={Jimbo, M.},
  isbn={9789812798336},
  series={Advanced series in mathematical physics},
  url={https://books.google.ie/books?id=MLjACwAAQBAJ},
  year={1990},
  publisher={World Scientific Publishing Company Pte Limited}
}

@article{de_Leeuw_2019,
   title={Classifying integrable spin-1/2 chains with nearest neighbour interactions},
   volume={52},
   ISSN={1751-8121},
   url={http://dx.doi.org/10.1088/1751-8121/ab529f},
   DOI={10.1088/1751-8121/ab529f},
   number={50},
   journal={Journal of Physics A: Mathematical and Theoretical},
   publisher={IOP Publishing},
   author={de Leeuw, Marius and Pribytok, Anton and Ryan, Paul},
   year={2019},
   month=nov, pages={505201} }
   
   @article{de_Leeuw_2021,
   title={Yang-Baxter and the Boost: splitting the difference},
   volume={11},
   ISSN={2542-4653},
   url={http://dx.doi.org/10.21468/SciPostPhys.11.3.069},
   DOI={10.21468/scipostphys.11.3.069},
   number={3},
   journal={SciPost Physics},
   publisher={Stichting SciPost},
   author={de Leeuw, Marius and Paletta, Chiara and Pribytok, Anton and Retore, Ana L. and Ryan, Paul},
   year={2021},
   month=sep }

   @article{Drinfeld:1985rx,
    author = "Drinfeld, V. G.",
    title = "{Hopf algebras and the quantum Yang-Baxter equation}",
    journal = "Sov. Math. Dokl.",
    volume = "32",
    pages = "254--258",
    year = "1985"
}

@article{Faddeev:1987ih,
    author = "Faddeev, L. D. and Reshetikhin, N. Yu. and Takhtajan, L. A.",
    title = "{Quantization of Lie Groups and Lie Algebras}",
    reportNumber = "LOMI-E-14-87",
    journal = "Alg. Anal.",
    volume = "1",
    number = "1",
    pages = "178--206",
    year = "1989"
}

@inproceedings{Faddeev:1996iy,
    author = "Faddeev, L. D.",
    title = "{How algebraic Bethe ansatz works for integrable model}",
    booktitle = "{Les Houches School of Physics: Astrophysical Sources of Gravitational Radiation}",
    eprint = "hep-th/9605187",
    archivePrefix = "arXiv",
    pages = "pp. 149--219",
    month = "5",
    year = "1996"
}
@article{Jimbo:1985zk,
    author = "Jimbo, Michio",
    title = "{A q difference analog of U(g) and the Yang-Baxter equation}",
    doi = "10.1007/BF00704588",
    journal = "Lett. Math. Phys.",
    volume = "10",
    pages = "63--69",
    year = "1985"
}

@article{gerstenhaber1997boundarysolutionsquantumyangbaxter,
      author="Murray Gerstenhaber and Anthony Giaquinto",
      title="{Boundary solutions of the quantum Yang-Baxter equation and solutions in three dimensions}", 
       year={1997},
      eprint={q-alg/9710033},
      archivePrefix={arXiv},
      primaryClass={q-alg},
      url={https://arxiv.org/abs/q-alg/9710033}, 
}

@article{KASHAEV_1993,
   title={GENERALIZED YANG-BAXTER EQUATION},
   volume={08},
   ISSN={1793-6632},
   url={http://dx.doi.org/10.1142/S0217732393003603},
   DOI={10.1142/s0217732393003603},
   number={24},
   journal={Modern Physics Letters A},
   publisher={World Scientific Pub Co Pte Lt},
   author={KASHAEV, R. M. and STROGANOV, YU. G.},
   year={1993},
   month=aug, pages={2299–2309} }
   
   @misc{tarasov1994solutionstwistedyangbaxterequation,
      title={On Solutions to the Twisted Yang-Baxter equation}, 
      author={Vitaly Tarasov},
      year={1994},
      eprint={hep-th/9403011},
      archivePrefix={arXiv},
      primaryClass={hep-th},
      url={https://arxiv.org/abs/hep-th/9403011}, 
}
@article{Etingof_1998,
   title={Geometry and Classificatin of Solutions of the Classical Dynamical Yang-Baxter Equation},
   volume={192},
   ISSN={1432-0916},
   url={http://dx.doi.org/10.1007/s002200050292},
   DOI={10.1007/s002200050292},
   number={1},
   journal={Communications in Mathematical Physics},
   publisher={Springer Science and Business Media LLC},
   author={Etingof, Pavel and Varchenko, Alexander},
   year={1998},
   month=mar, pages={77–120} }
   
@article{TakFaddeev,
doi = {10.1070/RM1979v034n05ABEH003909},
url = {https://dx.doi.org/10.1070/RM1979v034n05ABEH003909},
year = {1979},
month = {oct},
publisher = {},
volume = {34},
number = {5},
pages = {11},
author = {L. A. Takhtajan and  L. D. Faddeev},
title = {THE QUANTUM METHOD OF THE INVERSE PROBLEM AND THE HEISENBERG XYZ MODEL},
journal = {Russian Mathematical Surveys},
abstract = {CONTENTS Introduction   § 1. Classical statistical physics on a two-dimensional lattice and quantum mechanics on a chain   § 2. Connection with the inverse problem method   § 3. The six-vertex model   § 4. Generating vectors and permutation relations   § 5. The general Bethe Ansatz   § 6. Integral equations  Conclusion  Appendix 1  Appendix 2  References}
}

@article{Krichever1981BaxtersEA,
  title={Baxter's equations and algebraic geometry},
  author={Igor Moiseevich Krichever},
  journal={Functional Analysis and Its Applications},
  year={1981},
  volume={15},
  pages={92-103},
  url={https://api.semanticscholar.org/CorpusID:121076970}
}

@article{Padmanabhan:2024zma,
    author = "Padmanabhan, Pramod and Korepin, Vladimir",
    title = "{Solving the Yang-Baxter, tetrahedron and higher simplex equations using Clifford algebras}",
    eprint = "2404.11501",
    archivePrefix = "arXiv",
    primaryClass = "hep-th",
    doi = "10.1016/j.nuclphysb.2024.116664",
    journal = "Nucl. Phys. B",
    volume = "1007",
    pages = "116664",
    year = "2024"
}

@article{Dragovich,
    author = "Dragovich, V.I.",
    title = "{Solutions to the Yang equation with rational spectral curves}",
    journal = "Algebra i Analiz",
    volume = "4",
    pages = "104–116",
    year = "1992"
}

@book{Korepin_Bogoliubov_Izergin_1993, place={Cambridge}, series={Cambridge Monographs on Mathematical Physics}, title={Quantum Inverse Scattering Method and Correlation Functions}, publisher={Cambridge University Press}, author={Korepin, V. E. and Bogoliubov, N. M. and Izergin, A. G.}, year={1993}, collection={Cambridge Monographs on Mathematical Physics}} 

\end{bibtex}

\bibliographystyle{nb}
\bibliography{\jobname}

\end{document}